\documentclass[smallextended]{svjour3}
\usepackage[english]{babel}
\usepackage[cmex10]{amsmath}
\usepackage{amsfonts}
\usepackage[ruled, vlined, nofillcomment, linesnumbered]{algorithm2e}
\usepackage{xspace}
\usepackage{url}
\usepackage{icomma}
\usepackage{booktabs}
\usepackage{nicefrac}
\usepackage{tikz}
\usepackage{pgfplots}
\usepackage[numbers,sort&compress]{natbib}

\newcommand{\ifcondense}{\iftrue}

\newcommand{\set}[1]{\left\{#1\right\}}

\newcommand{\fpr}[1]{\mathopen{}\left(#1\right)}

\newcommand{\abs}[1]{{\left|#1\right|}}

\newcommand{\np}{\textbf{NP}}
\newcommand{\apx}{\textbf{APX}}
\newcommand{\naturals}{\mathbb{N}}
\newcommand{\integers}{\mathbb{Z}}
\newcommand{\funcdef}[3]{{#1}:{#2} \to {#3}}
\newcommand{\define}{\leftarrow}
\newcommand{\reals}{{\mathbb{R}}}

\DeclareRobustCommand{\dispfunc}[2]{%
  \ensuremath{%
  \ifthenelse{\equal{#2}{}}%
    {\mathit{#1}}%
    {\mathit{#1}\fpr{#2}}}}

\newcommand{\circulation}[1]{\dispfunc{circ}{#1}}
\newcommand{\cg}[1]{\dispfunc{H}{#1}}

\newcommand{\score}[1]{\dispfunc{q}{#1}}
\newcommand{\pen}[1]{\dispfunc{p}{#1}}

\newcommand{\penlin}[1]{\dispfunc{p_l}{#1}}
\newcommand{\pencons}[1]{\dispfunc{p_c}{#1}}
\newcommand{\pensum}[1]{\dispfunc{p_s}{#1}}

\newcommand{\bigO}[1]{\dispfunc{\mathcal{O}}{#1}}

\newcommand{\inback}[1]{\dispfunc{ib}{#1}}
\newcommand{\outback}[1]{\dispfunc{ob}{#1}}
\newcommand{\dbackone}[1]{\dispfunc{db_N}{#1}}
\newcommand{\dbacktwo}[1]{\dispfunc{db_P}{#1}}
\newcommand{\back}[1]{\dispfunc{b}{#1}}
\newcommand{\flux}[1]{\dispfunc{flux}{#1}}
\newcommand{\diff}[1]{\dispfunc{d}{#1}}
\newcommand{\adj}[1]{\dispfunc{adj}{#1}}
\newcommand{\opt}[1]{\dispfunc{opt}{#1}}
\newcommand{\lopt}[1]{\dispfunc{o}{#1}}
\newcommand{\gain}[1]{\dispfunc{gain}{#1}}

\newcommand{\circprb}{\textsc{Capacitated circulation}\xspace}
\newcommand{\uncircprb}{\textsc{Circulation}\xspace}
\newcommand{\fasprb}{\textsc{FAS}\xspace}
\newcommand{\mcutprb}{\textsc{Maximum Cut}\xspace}

\newcommand{\genagonyprb}{\textsc{Agony-with-shifts}\xspace}
\newcommand{\relief}{\textsc{Relief}\xspace}
\newcommand{\algrr}{\textsc{rr}\xspace}

\newcommand{\algcount}{\textsc{LeftSmaller}\xspace}
\newcommand{\algsplit}{\textsc{Split}\xspace}
\newcommand{\algleft}{\textsc{ConstructLeft}\xspace}
\newcommand{\algright}{\textsc{ConstructRight}\xspace}

\newcommand{\canon}{\textsc{canon}\xspace}

\newcommand{\dtname}[1]{\textsl{#1}}

\SetKwComment{tcpas}{\{}{\}}
\SetCommentSty{textnormal}
\SetArgSty{textnormal}
\SetKw{False}{false}
\SetKw{True}{true}
\SetKw{Null}{null}
\SetKwInOut{Output}{output}
\SetKwInOut{Input}{input}
\SetKw{AND}{and}
\SetKw{OR}{or}
\SetKw{Break}{break}

\SetKwFor{Until}{until}{do}{}

\pgfdeclarelayer{background}
\pgfdeclarelayer{foreground}
\pgfsetlayers{background,main,foreground}

\definecolor{yafaxiscolor}{rgb}{0.3, 0.3, 0.3}

\definecolor{yafcolor1}{rgb}{0.4, 0.165, 0.553}
\definecolor{yafcolor2}{rgb}{0.949, 0.482, 0.216}
\definecolor{yafcolor3}{rgb}{0.47, 0.549, 0.306}
\definecolor{yafcolor4}{rgb}{0.925, 0.165, 0.224}
\definecolor{yafcolor5}{rgb}{0.141, 0.345, 0.643}
\definecolor{yafcolor6}{rgb}{0.965, 0.933, 0.267}
\definecolor{yafcolor7}{rgb}{0.627, 0.118, 0.165}
\definecolor{yafcolor8}{rgb}{0.878, 0.475, 0.686}

\tikzstyle{exnode} = [inner sep = 1pt]
\tikzstyle{labnode} = [sloped, text = black, font = \scriptsize, inner sep = 1pt]
\tikzstyle{exedge} = [yafcolor5, draw, thick, >=latex, ->]
\tikzstyle{exedge2} = [yafcolor2, draw, thick, >=latex, ->]

\newlength{\yafaxispad}
\newlength{\yaftlpad}
\newlength{\yaflabelpad}
\newlength{\yafaxiswidth}
\newlength{\yafticklen}
\makeatletter
\def\pgfplots@drawtickgridlines@INSTALLCLIP@onorientedsurf#1{}
\makeatother

\newcommand{\yafdrawxaxis}[2]{
	\pgfplotstransformcoordinatex{#1}\let\xmincoord=\pgfmathresult 
	\pgfplotstransformcoordinatex{#2}\let\xmaxcoord=\pgfmathresult 
	\pgfsetlinewidth{\yafaxiswidth} 
	\pgfsetcolor{yafaxiscolor}
	\pgfpathmoveto{\pgfpointadd{\pgfpointadd{\pgfplotspointrelaxisxy{0}{0}}{\pgfqpointxy{\xmincoord}{0}}}{\pgfqpoint{-0.5\yafaxiswidth}{\yafaxispad}}}
	\pgfpathlineto{\pgfpointadd{\pgfpointadd{\pgfplotspointrelaxisxy{0}{0}}{\pgfqpointxy{\xmaxcoord}{0}}}{\pgfqpoint{0.5\yafaxiswidth}{\yafaxispad}}}
	\pgfusepath{stroke}

}
\newcommand{\yafdrawyaxis}[2]{
	\pgfplotstransformcoordinatey{#1}\let\ymincoord=\pgfmathresult 
	\pgfplotstransformcoordinatey{#2}\let\ymaxcoord=\pgfmathresult 
	\pgfsetlinewidth{\yafaxiswidth} 
	\pgfsetcolor{yafaxiscolor}
	\pgfpathmoveto{\pgfpointadd{\pgfpointadd{\pgfplotspointrelaxisxy{0}{0}}{\pgfqpointxy{0}{\ymincoord}}}{\pgfqpoint{\yafaxispad}{-0.5\yafaxiswidth}}}
	\pgfpathlineto{\pgfpointadd{\pgfpointadd{\pgfplotspointrelaxisxy{0}{0}}{\pgfqpointxy{0}{\ymaxcoord}}}{\pgfqpoint{\yafaxispad}{0.5\yafaxiswidth}}}
	\pgfusepath{stroke}
}

\newcommand{\yafdrawaxis}[4]{\yafdrawxaxis{#1}{#2}\yafdrawyaxis{#3}{#4}}

\pgfplotscreateplotcyclelist{yaf}{%
{yafcolor1,mark options={scale=0.75},mark=o}, 
{yafcolor2,mark options={scale=0.75},mark=square},
{yafcolor3,mark options={scale=0.75},mark=triangle},
{yafcolor4,mark options={scale=0.75},mark=o},
{yafcolor5,mark options={scale=0.75},mark=o},
{yafcolor6,mark options={scale=0.75},mark=o},
{yafcolor7,mark options={scale=0.75},mark=o},
{yafcolor8,mark options={scale=0.75},mark=o}} 

\pgfplotsset{axis y line=left, axis x line=bottom,
	tick align=outside,
	compat = 1.3,
	tickwidth=\yafticklen,
	clip = false,
	every axis title shift = 0pt,
    x axis line style= {-, line width = 0pt, opacity = 0},
    y axis line style= {-, line width = 0pt, opacity = 0},
    x tick style= {line width = \yafaxiswidth, color=yafaxiscolor, yshift = \yafaxispad},
    y tick style= {line width = \yafaxiswidth, color=yafaxiscolor, xshift = \yafaxispad},
    x tick label style = {font=\scriptsize, yshift = \yaftlpad},
    y tick label style = {font=\scriptsize, xshift = \yaftlpad},
    every axis y label/.style = {at = {(ticklabel cs:0.5)}, rotate=90, anchor=center, font=\scriptsize, yshift = -\yaflabelpad},
    every axis x label/.style = {at = {(ticklabel cs:0.5)}, anchor=center, font=\scriptsize, yshift = \yaflabelpad},
    x tick label style = {font=\scriptsize, yshift = 1pt},
    grid = major,
    major grid style  = {dash pattern = on 1pt off 3 pt},
	every axis plot post/.append style= {line width=\yafaxiswidth} ,
	legend cell align = left,
	legend style = {inner sep = 1pt, cells = {font=\scriptsize}},
	legend image code/.code={%
		\draw[mark repeat=2,mark phase=2,#1] 
		plot coordinates { (0cm,0cm) (0.15cm,0cm) (0.3cm,0cm) };%
	} 
}

\begin{document}

\title{Tiers for peers\thanks{The research described in this paper builds upon and extends the work appearing in \nobreak{ICDM15} as \cite{tatti:15:agony}.}}
\subtitle{a practical algorithm for discovering hierarchy in weighted networks}

\author{Nikolaj Tatti}
\institute{Nikolaj Tatti
\at Helsinki Institute for Information Technology (HIIT) and\\ Department of Information and Computer Science, Aalto University, Finland \\
\email{nikolaj.tatti@aalto.fi}
}

\maketitle


\begin{abstract} 
Interactions in many real-world phenomena can be explained by a strong
hierarchical structure. Typically, this structure or ranking is not known;
instead we only have observed outcomes of the interactions, and the goal is to
infer the hierarchy from these observations.
Discovering a hierarchy in the context of directed networks can be formulated as
follows: given a graph, partition vertices into levels such that, ideally,
there are only edges from upper levels to lower levels. The ideal case can only
happen if the graph is acyclic. Consequently, in practice we have to introduce
a penalty function that penalizes edges violating the hierarchy. A practical
variant for such penalty is agony, where each violating edge is penalized based
on the severity of the violation. Hierarchy minimizing agony can be discovered
in $\bigO{m^2}$ time, and much faster in practice. In this paper we introduce several
extensions to agony. We extend the definition for weighted graphs and allow a
cardinality constraint that limits the number of levels. While, these are
conceptually trivial extensions, current algorithms cannot handle them, nor
they can be easily extended. We solve the problem by showing the connection
to the capacitated circulation problem, and we demonstrate that we can compute
the exact solution fast in practice for large datasets. We also introduce a provably fast heuristic
algorithm that produces rankings with competitive scores.
In addition, we show that we can compute agony in polynomial
time for any convex penalty, and, to complete the picture, we show that minimizing
hierarchy with any concave penalty is an \np-hard problem.
\end{abstract}



\keywords{Hierarchy discovery; agony; capacitated circulation; weighed graphs}

\section{Introduction}\label{sec:intro}

Interactions in many real-world phenomena can be explained by a strong
hierarchical structure. As an example, it is more likely that a line manager in
a large, conservative company will write emails to her employees than the other
way around. Typically, this structure or ranking is not known; instead we only
have observed outcomes of the interactions, and the goal is to infer the
hierarchy from these observations. Discovering hierarchies or ranking has
applications in various domains:
(\emph{i})
ranking individual players or teams based on how well they play against
each other~\citep{elo1978rating},
(\emph{ii})
discovering dominant animals within a single herd, or ranking
species based on who-eats-who networks~\citep{jameson:99:behaviour},
(\emph{iii})
inferring hierarchy in work-places, such as, U.S. administration~\citep{DBLP:conf/cse/MaiyaB09},
(\emph{iv})
summarizing browsing behaviour~\citep{DBLP:conf/icdm/MacchiaBGC13},
(\emph{v})
discovering hierarchy in social networks~\citep{gupte:11:agony}, for example, if
we were to rank twitter users, the top-tier users would be the
content-providers, middle-tiers would spread the content, while the bottom-tier
are the consumers.

We consider the following problem of discovering hierarchy in the context of
directed networks: given a directed graph, partition vertices into
ranked groups such that there are only edges from upper groups to
lower groups. 

Unfortunately, such a partitioning is only possible when the input graph has no
cycles. Consequently, a more useful problem definition is to define a penalty
function $\pen{}$ on the edges. This function should penalize edges that are
violating a hierarchy. Given a penalty function, we are then asked to find the
hierarchy that minimizes the total penalty.

The feasibility of the optimization problem depends drastically on the choice
of the penalty function. If we attach a constant penalty to any edge that violates the hierarchy,
that is, the target vertex is ranked higher or equal than the
source vertex, then this problem corresponds to a feedback arc set problem, a
well-known \np-hard problem~\cite{dinur:05:cover}, even without a known constant-time approximation
algorithm~\cite{even:98:feedback}.

A more practical variant is to penalize the violating edges by the severity of
their violation. That is, given an edge $(u, v)$ we compare the ranks of the
vertices $r(u)$ and $r(v)$ and assign a penalty of $\max (r(u) - r(v) + 1, 0)$.
Here, the edges that respect the hierarchy receive a penalty of $0$, edges that
are in the same group receive a penalty of $1$, and penalty increases linearly
as the violation becomes more severe, see Figure~\ref{fig:toypen}. This particular score is referred as
\emph{agony}. Minimizing agony was introduced by~\citet{gupte:11:agony} where
the authors provide an exact $\bigO{nm^2}$ algorithm, where $n$ is the number of
vertices and $m$ is the number of edges. A faster discovery algorithm with the
computational complexity of $\bigO{m^2}$ was introduced by~\citet{tatti:14:agony}.
In practice, the bound $\bigO{m^2}$ is very pessimistic and we can compute agony
for large graphs in reasonable time.

In this paper we specifically focus on agony, and provide the following main
extensions for discovering hierarchies in graphs.

\textbf{weighted graphs:} We extend the notion of the agony to graphs with
weighted edges. Despite being a conceptually trivial extension, current
algorithms~\citep{gupte:11:agony,tatti:14:agony} for computing agony are
specifically design to work with unit weights, and cannot be used directly or
extended trivially. Consequently, we need a new approach to minimize the agony,
and in order to do so, we demonstrate that we can transform the problem into a
capacitated circulation, a classic graph task known to have a polynomial-time
algorithm.

\textbf{cardinality constraint:}
The original definition of agony does not restrict the number of groups in the
resulting partition. Here, we introduce a cardinality constraint $k$ and we are
asking to find the optimal hierarchy with at most $k$ groups. This constraint
works both with weighted and non-weighted graphs. Current algorithms for
solving agony cannot handle cardinality constraints. Luckily, we can
enforce the constraint when we transform the problem into a capacitated
circulation problem.

\textbf{fast heuristic:}
We introduce a fast divide-and-conquer heuristic. This heuristic is provably
fast, see Table~\ref{tab:algortime}, and---in our experiments---produces competitive scores when
compared to the optimal agony.

\textbf{convex edge penalties:} Minimizing agony uses linear penalty for edges.
We show that if we replace the linear penalty with a convex
penalty, see Figure~\ref{fig:toypen}, we can still solve the problem in polynomial time by the capacitated circulation solver.
However, this extension increases the computational complexity.

\textbf{concave edge penalties:} To complete the picture, we also study concave
edge penalties, see Figure~\ref{fig:toypen}. We show that in this case discovering the optimal hierarchy is
an \np-hard problem. This provides a stark difference between concave and
convex edge penalties.

\textbf{canonical solution:} A hierarchy minimizing agony may not be unique.
For example, given a DAG any topological sorting of vertices will give you an
optimal agony of 0. To address this issue we propose to compute a
\emph{canonical} solution, where, roughly speaking, the vertices are ranked as
high as possible without compromising the optimality of the solution.  We
demonstrate that this solution is unique, it creates a hierarchy with the least
amount of groups, and that we can compute it in $\bigO{n \log n + m}$ time, if we are provided
with the optimal solution and the flow resulted from solving the capacitated circulation.

\begin{figure}[ht!]
\begin{center}
\begin{tikzpicture}
\begin{axis}[cycle list name=yaf,
	xmin = -3, ymin = 0,
	width = 8cm, height = 4cm,
	ylabel = {edge penalty},
	xlabel = {rank difference, $r(u) - r(v)$},
	xtick = {-3,...,5},
	no markers]
\addplot+[yafcolor1] coordinates {(-3, 0) (0, 0) (0, 1)};
\addplot+[domain=0:5, yafcolor1] {1};
\addplot+[domain=0:5, yafcolor2] {x + 1};
\addplot+[domain=0:5, yafcolor3] {sqrt(x)*1.5 + 1};
\addplot+[domain=0:5, yafcolor4] {x*x/4 + 1};

\node[coordinate, label={[font = \scriptsize, inner sep = 1pt]right:linear = agony (polynomial)}] at (axis cs:5,6)  {};
\node[coordinate, label={[font = \scriptsize, inner sep = 1pt]right:concave (\np-hard)}] at (axis cs:5,4.35)  {};
\node[coordinate, label={[font = \scriptsize, inner sep = 1pt]right:convex (polynomial)}] at (axis cs:5,7.25)  {};
\node[coordinate, label={[font = \scriptsize, inner sep = 1pt]right:constant = FAS (\np-hard)}] at (axis cs:5,1)  {};

\pgfplotsextra{\yafdrawaxis{-3}{5}{0}{7.25}}
\end{axis}
\end{tikzpicture}
\end{center}

\caption{A toy example of edge penalties as a function of the rank difference between the vertices.}
\label{fig:toypen}
\end{figure}
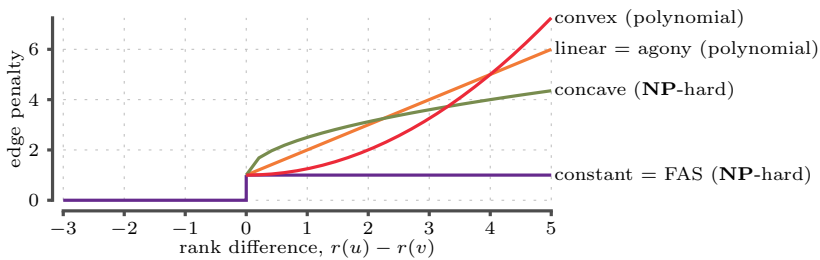

\begin{table}[ht!]
\caption{Summary of running times of different algorithms for computing agony:
$n$ is the number of vertices, $m$ is the number of edges, $k$ is the number of
allowed ranks.}
\label{tab:algortime}
\begin{tabular*}{\textwidth}{@{\extracolsep{\fill}}lrrr}
\toprule
Algorithm & variant & input type & running time \\
\midrule
Exact & plain &  & $\bigO{m \log n(m + n \log n)}$\\
Exact & speed-up &  unweighted  & $\bigO{m(\min(kn, m) + n \log n)}$\\
Exact & speed-up &  weighted  & $\bigO{m \log n(m + n \log n)}$ \\
Canonical & -- & optimal rank and the flow & $\bigO{m + n\log n}$ \\[1mm]
Heuristic & plain & no cardinality constraint & $\bigO{m \log n}$\\
Heuristic & plain & cardinality constraint & $\bigO{m \log n + k^2 n}$\\
Heuristic & SCC & no cardinality constraint & $\bigO{m \log n}$\\
Heuristic & SCC & cardinality constraint & $\bigO{m \log n + k^2n + km \log n}$\\
\bottomrule
\end{tabular*}
\end{table}

This paper is an extension of a conference paper~\citep{tatti:15:agony}. In
this extension we significantly speed-up the exact algorithm, propose a
provably fast heuristic, and provide a technique for selecting unique canonical
solutions among the optimal rankings.

The rest of the paper is organized as follows. We introduce the notation and
formally state the optimization problem in Section~\ref{sec:prel}. In
Section~\ref{sec:algo} we transform the optimization problem into a capacitated
circulation problem, allowing us a polynomial-time algorithm, and provide a speed-up in Section~\ref{sec:speedup}. In
Section~\ref{sec:score} we discuss alternative edge penalties. We demonstrate
how to extract a canonical optimal solution in Section~\ref{sec:canon}.
We discuss the related work in Section~\ref{sec:related} and present
experimental evaluation in Section~\ref{sec:exp}. Finally, we conclude the paper with
remarks in Section~\ref{sec:conclusions}.

\section{Preliminaries and problem definition}\label{sec:prel}


We begin with establishing preliminary notation and then defining the main
problem.

The main input to our problem is a \emph{weighted directed graph} which we will
denote by $G = (V, E, w)$, where $w$ is a function mapping an edge to a real
positive number. If $w$ is not provided, we assume that each edge has a weight
of 1. We will often denote $n = \abs{V}$ and $m = \abs{E}$.

As mentioned in the introduction, our goal is to partition vertices $V$.  We
express this partition with a \emph{rank assignment} $r$, a function
mapping a vertex to an integer. To obtain the groups from the rank assignment
we simply group the vertices having the same rank.

Given a graph $G = (V, E)$ and a rank assignment $r$, we will say that an edge
$(u, v)$ is \emph{forward} if $r(u) < r(v)$, otherwise edge is
\emph{backward}, even if $r(u) = r(v)$.
Ideally, rank assignment $r$ should not have backward edges, that is, for any
$(u, v) \in E$ we should have $r(u) < r(v)$. However, this is only possible
when $G$ is a DAG. For a more general case, we assume that we are given a
penalty function $\pen{}$, mapping an integer to a real number. The penalty
for a single edge $(u, v)$ is then equal to $\pen{d}$, where $d = r(u) - r(v)$. 
If $\pen{d} = 0$, whenever $d < 0$, then the forward edges will receive $0$ penalty.

We highlight two penalty functions. The first one assigns a constant penalty
to each backward edge,
\[
	\pencons{d} =
	\begin{cases}
	1 & \text{ if } d \geq 0 \\
	0 & \text{ otherwise } \quad.
	\end{cases}
\]
The second penalty function assigns a linear penalty to each backward edge,
\[
	\penlin{d} = \max(0, d + 1)\quad.
\]
For example,
an edge $(u, v)$ with $r(u) = r(v)$ is penalized by $\penlin{r(u) - r(v)} = 1$, the penalty is equal to $2$ if $r(u) = r(v) + 1$, and so on.

Given a penalty function and a rank assignment we can now define the 
the score for the ranking to be the sum of the weighted penalties.
\begin{definition}
Assume a weighted directed graph $G = (V, E, w)$ and a rank assignment $r$.
Assume also a cost function $\pen{}$ mapping an integer to a real number.
We define a score for a rank assignment to be
\[
	\score{G, r, \pen{}} = \sum_{e = (u, v) \in E} w(e) \pen{r(u) - r(v)}\quad.
\]
\end{definition}
We will refer the score $\score{G, r, \penlin{}}$ as \emph{agony}.

\begin{example}

Consider the left ranking $r_1$ of a graph $G$ given in Figure~\ref{fig:ex}.
This ranking has 5 backward edges, consequently, the penalty is $\score{G, r_1,
\pencons{}} = 5$. On the other hand, there are 2 edges, $(i, a)$ and $(e, g)$, having the agony of 1.
Moreover, 2 edges has agony of 2 and $(d, b)$ has agony of 3. Hence, agony is
equal to
\[
	\score{G, r_1, \penlin{}} = 2\times 1 + 2 \times 2 + 1 \times 3  = 10\quad.
\]
The agony for the right ranking $r_2$ 
is
$\score{G, r_2, \penlin{}} = 7$. Consequently, $r_2$ yields a better ranking in terms of 
agony.

\begin{figure}[ht!]

\newlength{\levelsep}
\newlength{\nodesep}
\setlength{\levelsep}{-0.6cm}
\setlength{\nodesep}{0.7cm}

\tikzstyle{node} = [fill = white, circle, inner sep = 1pt]
\tikzstyle{label} = [inner sep = 0pt]

\tikzstyle{dagedge} = [yafcolor5, -latex, thick]
\tikzstyle{backedge} = [yafcolor4, -latex, thick, densely dotted]
\tikzstyle{leveledge} = [yafaxiscolor!50, dashed]

\hfill
\begin{tikzpicture}

\draw[leveledge] (-0.5\nodesep, 0) -- (4.5\nodesep, 0);
\draw[leveledge] (-0.5\nodesep, \levelsep) -- (4.5\nodesep, \levelsep);
\draw[leveledge] (-0.5\nodesep, 2\levelsep) -- (4.5\nodesep, 2\levelsep);
\draw[leveledge] (-0.5\nodesep, 3\levelsep) -- (4.5\nodesep, 3\levelsep);

\node[node] (u0) at (\nodesep, 0) {$a$};
\node[node] (u1) at (0, \levelsep) {$b$};
\node[node] (u2) at (\nodesep, 2\levelsep) {$c$};
\node[node] (u3) at (0, 3\levelsep) {$d$};

\node[node] (u4) at (2\nodesep, 1\levelsep) {$e$};
\node[node] (u5) at (3\nodesep, 0\levelsep) {$f$};
\node[node] (u6) at (4\nodesep, 1\levelsep) {$g$};
\node[node] (u7) at (3\nodesep, 2\levelsep) {$h$};
\node[node] (u9) at (0\nodesep, 0\levelsep) {$i$};

\draw[dagedge] (u0) edge (u2);
\draw[dagedge] (u2) edge (u3);
\draw[backedge] (u3) edge (u1);
\draw[backedge] (u1) edge (u0);
\draw[backedge] (u9) edge (u0);

\draw[dagedge] (u5) edge (u4);
\draw[backedge] (u4) edge (u6);
\draw[backedge] (u6) edge (u5);
\draw[dagedge] (u6) edge (u7);
\draw[dagedge] (u4) edge (u7);

\draw[dagedge] (u0) edge (u4);
\end{tikzpicture}\hfill
\begin{tikzpicture}

\draw[leveledge] (-0.5\nodesep, 0) -- (4.5\nodesep, 0);
\draw[leveledge] (-0.5\nodesep, \levelsep) -- (4.5\nodesep, \levelsep);
\draw[leveledge] (-0.5\nodesep, 2\levelsep) -- (4.5\nodesep, 2\levelsep);
\draw[leveledge] (-0.5\nodesep, 3\levelsep) -- (4.5\nodesep, 3\levelsep);

\node[node] (u0) at (\nodesep, \levelsep) {$a$};
\node[node] (u1) at (0, \levelsep) {$b$};
\node[node] (u2) at (\nodesep, 2\levelsep) {$c$};
\node[node] (u3) at (0, 3\levelsep) {$d$};
\node[node] (u9) at (0\nodesep, 0\levelsep) {$i$};

\node[node] (u4) at (2\nodesep, 2\levelsep) {$e$};
\node[node] (u5) at (2.5\nodesep, 1\levelsep) {$f$};
\node[node] (u6) at (4\nodesep, 1\levelsep) {$g$};
\node[node] (u7) at (3\nodesep, 3\levelsep) {$h$};

\draw[dagedge] (u0) edge (u2);
\draw[dagedge] (u2) edge (u3);
\draw[backedge] (u3) edge (u1);
\draw[backedge] (u1) edge (u0);
\draw[dagedge] (u9) edge (u0);

\draw[dagedge] (u5) edge (u4);
\draw[backedge] (u4) edge (u6);
\draw[backedge] (u6) edge (u5);
\draw[dagedge] (u6) edge (u7);
\draw[dagedge] (u4) edge (u7);

\draw[dagedge] (u0) edge (u4);
\end{tikzpicture}\hspace*{\fill}

\caption{Toy graphs. Backward edges are represented by dotted lines, while the
forward edges are represented by solid lines.  Ranks are represented by dashed
grey horizontal lines.}
\label{fig:ex}
\end{figure}
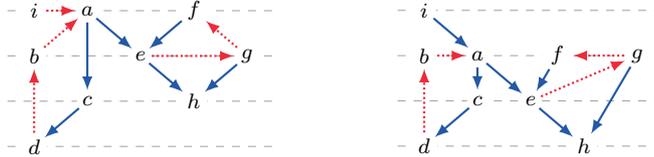
\end{example}

We can now state our main optimization problem.

\begin{problem}
\label{prb:opt}
Given a graph $G = (V, E, w)$, a cost function $\pen{}$, and an integer $k$,
find a rank assignment $r$ minimizing $\score{r, G}$
such that $0 \leq r(v) \leq k - 1$ for every $v \in V$.
We will denote the optimal score by $\score{G, k, \pen{}}$.
\end{problem}

We should point out that we have an additional constraint by demanding that the
rank assignment may have only $k$ distinct values, that is, we want to find at
most $k$ groups. Note that if we assume that the penalty function is non-decreasing
and does not penalize the forward edges, then setting $k = \abs{V}$ is equivalent
of ignoring the constraint.  This is the case since there are at most $\abs{V}$ groups
and we can safely assume that these groups obtain consecutive ranks. However,
an optimal solution may have less than $k$ groups, for example, if $G$ has no edges and we use $\penlin{}$ (or $\pencons{}$), then
a rank assigning each vertex to $0$ yields the optimal score of $0$.
We should also point out that if using $\pencons{}$, there is always an optimal solution
where each vertex has its own rank. This is not the case for agony.

It is easy to see that minimizing $\score{G, \pencons{}}$ is equivalent to
finding a directed acyclic subgraph with as many edges as possible. This is
known as \textsc{Feedback Arc Set} (\fasprb) problem, which is
\np-complete~\cite{dinur:05:cover}.

On the other hand, if we assume that $G$ has unit weights, and
set $k = \abs{V}$, then minimizing agony has a polynomial-time $\bigO{m^2}$ algorithm~\cite{gupte:11:agony,tatti:14:agony}.

\section{Computing agony}\label{sec:algo}
In this section we present a technique for minimizing agony, that is, solving
Problem~\ref{prb:opt} using $\penlin{}$ as a penalty.  In order to do this we show
that this problem is in fact a dual problem of the known graph problem, closely
related to the minimum cost max-flow problem.

\subsection{Agony with shifts}
\label{sec:shift}

We begin with an extension to our optimization problem.

\begin{problem}[\genagonyprb]
Given a graph $G = (V, E, w, s)$, where $w$ maps an edge to a, possibly infinite,
non-negative value, and $s$ maps an edge to a possibly negative integer,
find a rank assignment $r$ minimizing
\[
	\sum_{e = (u, v) \in E} w(e) \times \max(r(u) - r(v) + s(e), 0)\quad.
\]
We denote the optimal sum with $\score{G}$.
\end{problem}

In order to transform the problem of minimizing agony to \genagonyprb,
assume a graph $G = (V, E, w)$ and an integer $k$. 
We define a graph $H = (W, F, w, s)$ as follows.
The vertex set $W$ consists of 2 groups:
\emph{(i)} $\abs{V}$ vertices, each vertex corresponding to a vertex in $G$
\emph{(ii)} $2$ additional vertices $\alpha$ and $\omega$.
For each edge $e = (u, v) \in E$, we add an edge $f = (u, v)$ to $F$.
We set $w(f) = w(e)$ and $s(f) = 1$.
We add edges $(v, \omega)$ and $(\alpha, v)$ for every $v \in V$ with
$s(v, \omega) = s(\alpha, v) = 0$ and
$w(v, \omega) = w(\alpha, v) = \infty$.
Finally we add $(\omega, \alpha)$ with $s(\omega, \alpha) = 1 - k$ and $w(\omega, \alpha) = \infty$.
We will denote this graph by $\cg{G, k} = H$.

\begin{example}
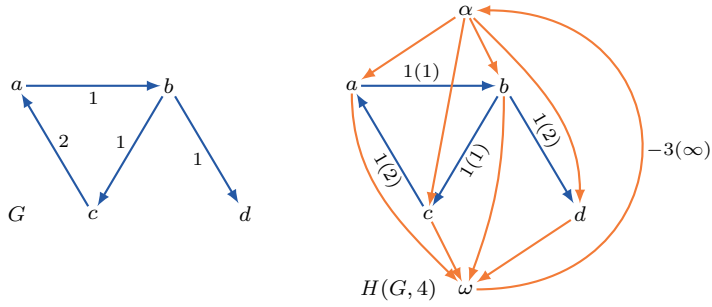
\begin{figure}[ht!]
\hspace*{\fill}
\begin{tikzpicture}[baseline=0]
\node at( 0, -1.7) {$G$};
\node[exnode] at (0, 0) (a) {$a$};
\node[exnode] at (2, 0) (b) {$b$};
\node[exnode] at (1, -1.7) (c) {$c$};
\node[exnode] at (3, -1.7) (d) {$d$};

\draw (a) edge[exedge] node[auto = right, black, font = \scriptsize, circle, inner sep = 1pt] {1} (b);
\draw (b) edge[exedge] node[auto = right, black, font = \scriptsize, circle, inner sep = 1pt] {1} (c);
\draw (c) edge[exedge] node[auto = right, black, font = \scriptsize, circle, inner sep = 1pt] {2} (a);
\draw (b) edge[exedge] node[auto = right, black, font = \scriptsize, circle, inner sep = 1pt] {1} (d);

\end{tikzpicture}\hfill
\begin{tikzpicture}[baseline=0]
\node[anchor = west] at(0, -2.7) {$\cg{G, 4}$};
\node[exnode] at (0, 0) (a) {$a$};
\node[exnode] at (2, 0) (b) {$b$};
\node[exnode] at (1, -1.7) (c) {$c$};
\node[exnode] at (3, -1.7) (d) {$d$};

\node[exnode] at (1.5, 1) (s) {$\alpha$};
\node[exnode] at (1.5, -2.7) (t) {$\omega$};

\draw (a) edge[exedge] node[labnode, auto = left, pos = 0.45] {$1(1)$} (b);

\draw (b) edge[exedge] node[labnode, auto= left, pos = 0.7] {$1(1)$} (c);

\draw (c) edge[exedge] node[labnode, auto = left, pos = 0.2] {$1(2)$} (a);

\draw (b) edge[exedge] node[labnode, auto = left, pos = 0.2] {$1(2)$} (d);

\draw (a) edge[exedge2, out = -90] (t);
\draw (b) edge[exedge2, out = -90, in = 70] (t);
\draw (c) edge[exedge2] (t);
\draw (d) edge[exedge2] (t);

\draw (s) edge[exedge2] (a);
\draw (s) edge[exedge2] (b);
\draw (s) edge[exedge2] (c);
\draw (s) edge[exedge2, out = -45, in = 90] (d);

\draw (t) edge[exedge2, out = 0, in = 0, looseness = 2] node[auto = right, black, font = \scriptsize, circle, inner sep = 1pt] {$-3(\infty)$} (s);

\end{tikzpicture}\hspace*{\fill}
\caption{Toy graph $G$ and the related circulation graph $\cg{G, 4}$. Edge costs and shifts for $(\alpha, v)$
and $(v, \omega)$ are omitted to avoid clutter.}
\label{fig:toycirculation}
\end{figure}
Consider $G = (V, E)$, a graph with $4$ vertices and $4$ edges, given in Figure~\ref{fig:toycirculation}. 
Set cardinality constraint $k = 4$.  In order to construct $\cg{G, k}$ we add 
two additional vertices $\alpha$ and $\omega$ to enforce the cardinality constraint $k$. We
set edge costs to $-1$ and edges capacities to be the weights of the input graph.
We connect $\alpha$ and $\omega$ with $a$, $b$, $c$, and $d$,
and finally we connect $\omega$ to $\alpha$. The resulting graph is given in Figure~\ref{fig:toycirculation}.
\end{example}

\subsection{Agony is a dual problem of Circulation}

Minimizing agony is closely related to a circulation problem, where the goal is
to find a circulation with a minimal cost satisfying certain balance equations.

\begin{problem}[\circprb]
\label{prb:circ}
Given a graph $G = (V, E, c, s)$, where $c$ maps an edge to a, possibly infinite,
non-negative value, and $s$ maps an edge to a possibly negative integer,
find a flow $f$ such
that $0 \leq f(e) \leq c(e)$ for every $e \in E$ and
\[
	\sum_{e = (v, u) \in E} f(e) = \sum_{e = (u, v) \in E} f(e), \quad\text{for every } v \in V
\]
maximizing
\[
	\sum_{e \in E} s(e)f(e)\quad.
\]
We denote the above sum as $\circulation{G}$.
\end{problem}

This problem is known as capacitated circulation problem, and can be solved in
$\bigO{m \log n(m + n \log n)}$ time with an algorithm presented
by~\citet{orlin:93:flow}.  We should stress that we allow $s$ to be negative.
We also allow capacities for certain edges to be infinite,
which simply means that $f(e) \leq c(e)$ is not enforced, if $c(e) = \infty$.

The following proposition shows the connection between the agony and the
capacitated circulation problem. 

\begin{proposition}
\label{prop:lp}
Assume a weighted directed graph with shifts, $G = (V, E, w, s)$.
Then $\score{G} = \circulation{G}$.
\end{proposition}

\ifcondense
\begin{proof}
Let $G = (V, E, w, s)$. 
To prove this result we will show that computing $\circulation{G}$
is a linear program, whose dual corresponds to optimizing \genagonyprb.
In order to do this, we first express a general \circprb problem as a linear program,
\begin{align*}
	\text{maximize } &\sum_{(u, v) \in E} s(u, v)f(u, v) & \text{such that}&\\
	\sum_{(v, u) \in V} f(v, u) & =  \sum_{(u, v) \in V} f(u, v) , & \text{ for every } v \in V&,\\
	w(u, v) & \geq f(u, v) \geq 0, & \text{ for every } (u, v) \in E&\quad.\\
\end{align*}

This program has the following dual program,
\begin{align}
	\text{minimize } & \sum_{(u, v) \in E} \eta(u, v)w(u, v)\hspace{-1cm}  \nonumber\\
	\text{such that } & \text{for every } (u, v) \in E\hspace{-1cm}\nonumber\\
	\pi(v) - \pi(u) + \eta(u, v) & \geq s(u, v), & \text{ if } w(u, v) < \infty,\nonumber\\
	\pi(v) - \pi(u) & \geq s(u, v), & \text{ if } w(u, v)  = \infty,\nonumber\\
	\eta(u, v) & \geq 0,& \label{eq:dual}
\end{align}
which is optimized over the variables $\pi$ and $\eta$.

If $\pi$ are integers, then they correspond to the ranking $r$.
Moreover, $\eta(u, v) = \max(\pi(u) - \pi(v) + s(u, v), 0)$.
So that, $w(u, v)\eta(u, v)$ corresponds to the penalty term in the sum of \genagonyprb,
and the objective function of the dual program corresponds exactly to the objective
of \genagonyprb.

To complete the proof we need to show that there is an optimal integer-valued dual solution $\pi$ and $\eta$.
This result follows from the fact that the
constraints of the dual form an arc-vertex incidence matrix, which is
known to be totally unimodular~\cite[Corollary of Theorem~13.3]{papadimitriou:82:opt}, Since $s(u, v)$ are integers, Theorem~13.2~in~\cite{papadimitriou:82:opt}
implies that there is an optimal solution with integer-valued $\pi$, completing the proof.\qed
\end{proof}
\fi

\subsection{Algorithm for minimizing agony}

Proposition~\ref{prop:lp} states that we can compute agony but it does not
provide direct means to discover an optimal rank assignment. However, a closer
look at the proof reveals that 
minimizing agony is a dual problem of \circprb. That is,
if we were to solve the dual optimization problem given in Equation~\ref{eq:dual},
then we can extract the optimal ranking from the dual parameters $\pi$
by setting $r(v) = \pi(v) -
\pi(\alpha)$ for $v \in V$, where $\alpha$ is the special vertex added during the construction of $H$.

Luckily, the algorithms for solving \circprb by~\citet{edmonds:72:flow} or
by~\citet{orlin:93:flow} in fact solve Equation~\ref{eq:dual} and are guaranteed
to have integer-valued solution as long as the capacities $s(u, v)$ are
integers, which is the case for us.

If we are not enforcing the cardinality constraint, that is,  we are solving
$\score{G, k}$ with $k = \abs{V}$, we can obtain a significant speed-up by
decomposing $G$ to strongly connected components, and solve ranking for
individual components.

\begin{proposition}
\label{prop:decomp}
Assume a graph $G$, and set $k = \abs{V}$. Let $\set{C_i}$ be the strongly connected
components of $G$, ordered in a topological order. Let $r_i$ be the ranking
minimizing $\score{G(C_i), \abs{C_i}}$.
Let $b_i = \sum_{j = 1}^{i - 1} \abs{C_j}$. Then the ranking 
$r(v) = r_i(v) + b_i$, where $C_i$ is the component containing $v$, yields the optimal score $\score{G, k}$.
\end{proposition}

\begin{proof}
Note that $\max r(v) \leq k$, hence $r$ is a valid ranking.
Let $r'$ be the ranking minimizing $\score{G, k}$. Let $r'_i$ be the projection
of the ranking to $C_i$. Then
\[
	\score{G, r'} \geq \sum_{i = 1} \score{G(C_i), r'_i} \geq \sum_{i = 1} \score{G(C_i), r_i} = \score{G, r},
\]
where the last equality holds because any cross-edge between the components is a forward edge.
\qed
\end{proof}

\section{Speeding up the circulation solver}
\label{sec:speedup}

In this section we propose a modification to the circulation solver.  This
modification provides us with a modest improvement in computational complexity,
and---according to our experimental evaluation---significant improvement in
running time in practice. 

Before explaining the modification, we first need to revisit the original
Orlin's algorithm. We refer the reader to~\citep{orlin:93:flow} for a complete expose.

The solver actually solves a slightly different problem, namely, an
uncapacitated circulation.

\begin{problem}[\uncircprb]
\label{prb:uncirc}
Given a directed graph $F = (W, A, t, b)$ with weights on edges and biases on vertices, find a flow $f$ such
that $0 \leq f(e)$ for every $e \in A$ and
\begin{equation}
\label{eq:primcond}
    \sum_{(v, u) \in A} f(v, u) - \sum_{(u, v) \in A} f(u, v) = b(v), \quad\text{for every } v \in W
\end{equation}
minimizing
\[
    \sum_{(u, v) \in A} t(u, v)f(u, v)\quad.
\]
\end{problem}

To map our problem to \uncircprb, we us the trick described by~\citet{orlin:93:flow}:
we replace each capacitated edge $e = (v, w)$
with a vertex $u$ and two edges $(v, u)$ and $(w, u)$. We set $b(u) = -c(e)$,
and add $c(v, w)$ to $b(w)$. The costs are set to $t(v, u) = \max(-s(e), 0)$ and $t(w, u) = \max(s(e), 0)$.
For each uncapacitated edge $(v, w)$, we connect $v$ to $w$ with $t(v, w) = -s(v, w)$.\!\footnote{The reason for the minus sign is
that we expressed \uncircprb as a minimization problem and \circprb as a maximization problem.}

From now on, we will write $H = \cg{G, k}$, and $F = (W, A, s, b)$ to be the
graph modified as above.  We split $W$ to $W_1$ and $W_2$: $W_1$ are the
original vertices in $H$, while $W_2$ are the vertices rising from the
capacitated edges.

We also write $n$ and $m$ to be the number of vertices and edges in $H$,
respectively, and $n'$ and $m'$ to be the number of vertices and edges in $F$,
respectively.  Note that $n', m' \in \bigO{m}$.

Consider the dual of uncapacitated circulation. 
\begin{problem}[dual to \uncircprb]
\label{prb:dualuncirc}
Given a directed graph $F = (W, A, s, b)$ with weights on edges and biases on vertices, find dual variables $\pi$ on vertices 
maximizing
\[
    \sum_{(u, v) \in A} b(v)\pi(v)
\]
such that 
\begin{equation}
\label{eq:dualcond}
	t(e) + \pi(w) - \pi(v) \geq 0, \quad\text{for every } e = (v, w) \in A\quad.
\end{equation}
\end{problem}

The standard linear programming theory states that $f$ and $\pi$
satisfying Eq.~\ref{eq:primcond}--\ref{eq:dualcond}
are optimal solutions to
their respective problems if and only if the slackness conditions hold,
\begin{equation}
\label{eq:slack}
	(t(e) + \pi(w) - \pi(v))f(e) = 0, \quad\text{for every}\quad e = (v, w) \in A\quad.
\end{equation}

The main idea behind Orlin's algorithm is to maintain a flow $f$ and a dual $\pi$
satisfying Eqs.~\ref{eq:dualcond}--\ref{eq:slack}, and then iteratively enforce
Eq.~\ref{eq:primcond}. More specifically, we first define an \emph{excess} of a vertex to be
\[
	e(v) = b(v) + \sum_{(w, v) \in A} f(w, v) - \sum_{(v, w) \in A} f(v, w)\quad.
\]
Our goal is to force $e(v) = 0$ for every $v$. This is done in gradually in multiple iterations.
Assume that we are given $\Delta$, a granularity which we will use to modify the flow.
The following steps are taken:
\emph{(i)}
We first
construct a residual graph $R$ which consists of all the original edges, and reversed edges for all
edges with positive flow.
\emph{(ii)}
We then select a source $s$ with $e(s) \geq \alpha \Delta$\footnote{Here $\alpha$ is a fixed parameter $1/2 < \alpha < 1$, we use $\alpha = 3/4$.},
and construct a shortest path tree $T$ in $R$, weighted by $t(e) + \pi(w) - \pi(v)$, for $e = (v, w)$.
\emph{(iii)}
The dual variables are updated to $\pi(v) - d(v)$, where $d$ is the shortest distance from $s$ to $v$.
\emph{(iv)}
We select a sink $r$ with $e(r) \leq -\alpha \Delta$, and augment the flow along the path in $T$ from $r$ to $s$.
This is repeated until there are no longer viable options for $s$ or $r$.
After that we half $\Delta$, and repeat.

To guarantee polynomial convergence, we also must contract edges for which
$f(e) \geq 3n'\Delta$, where $n'$ is the number of vertices in the (original) input
graph.  Assume that we contract $(v, w)$ into a new vertex $u$. We set $\pi(u) =
\pi(v)$, $b(u) = b(v) + b(w)$.  We delete the edge $(v, w)$, and edges adjacent
to $v$ and $w$ are migrated to $u$; the cost of an edge $t(w, x)$ must be
changed to $t(w, x) + t(v, w)$, and similarly the cost of an edge $t(x, w)$
must be changed to $t(x, w) - t(v, w)$.
A high-level pseudo-code is given in Algorithm~\ref{alg:orlin}.

\begin{algorithm}[ht!]
\caption{Orlin's algorithm for solving \uncircprb.}
\label{alg:orlin}
	$\Delta \define \min(\max e(v), \max -e(v))$\;
	\While {there is excess} {
		contract any edges with $f(e) \geq 3n\Delta$\;
		\While {$\max  e(v) \geq \alpha \Delta$ \AND $\min e(v) \leq -\alpha \Delta$}  {
			$s \define $  a vertex with $e(s) \geq \alpha \Delta$\;
			$r \define $  a vertex with $e(r) \leq -\alpha \Delta$\;
			$T \define $ shortest path tree from $s$ in residual graph, weighted by $t(e) + \pi(w) - \pi(v)$\;
			update $\pi$ using $T$\;
			$P \define $ path in $T$ from $r$ to $s$\;
			augment flow by $\Delta$ along $P$\;
		}
		$\Delta \define \Delta / 2$\;
	}
\end{algorithm}

The bottleneck of this algorithm is computing the shortest path tree. This is
the step that we will modify. In order to do this we first point out that Orlin's algorithm
relies on two things that inner loop should do:
\emph{(i)} Eqs.~\ref{eq:dualcond}--\ref{eq:slack} must be maintained, and
\emph{(ii)} path augmentations are of granularity $\Delta$, after the augmentations
there should not be a viable vertex for a source or a viable vertex for a sink.
As long as these two conditions are met, the correctness proof given by~\citet{orlin:93:flow} holds.

Our first modification is instead of selecting one source $s$, we select
\emph{all} possible sources, $S \define \set{v \in V \mid e(v) \geq \alpha
\Delta}$, and compute the shortest path tree using $S$ as roots.  Once this tree
is computed, we subtract the shortest distance from $\pi$, select a sink $t$, and augment flow along the path from $t$ to
some root $s \in S$.

The following lemma guarantees that Eqs.~\ref{eq:dualcond}--\ref{eq:slack} are
maintained when $f$ and $\pi$ are modified.

\begin{lemma}
\label{lem:dualupdate}
Let $f$ and $\pi$ be flow and dual variables satisfying the slackness
conditions given in Eq.~\ref{eq:slack}. Let $S$ be a set of vertices. Define
$d(v)$ be the shortest distance from $S$ to $v$ in the residual graph with
weighted edges $t(e) + \pi(w) - \pi(v)$.  Let $\pi' = \pi - d$. Then $\pi'$
satisfy Eq.~\ref{eq:dualcond}, and $f$ and $\pi'$ respect the slackness
conditions in Eq.~\ref{eq:slack}. Moreover, $t(e) + \pi'(w) - \pi'(v) = 0$ for every edge in
the shortest path tree.
\end{lemma}

Note that since we modify $f$ only along the edges of the shortest path tree,
this lemma guarantees that Eq.~\ref{eq:slack} is also 
maintained when we augment $f$. The proof of this lemma is essentially the same as the
single-source version given by~\citet{orlin:93:flow}.

\begin{proof}
Let $e = (v, w) \in E$. Then $e$ is also in residual graph, and $d(w) \leq d(v) + t(e) + \pi(w) - \pi(v)$.
This implies
\[
	t(e) + \pi'(w) - \pi'(v) = t(e) + \pi(w) - d(w) - \pi(v) + d(v) \geq 0,
\]
proving the first claim. If $e$ is in the shortest path tree, then $d(w) = d(v) + t(e) + \pi(w) - \pi(v)$,
which implies the third claim, $t(e) + \pi(w) - \pi(v) = 0$.

To prove the second claim,
if $f(e) > 0$, then $t(e) + \pi(w) - \pi(v) = 0$.
Since $(w, v) = e'$ is also in residual graph, we must have $d(v) = d(w)$.
Thus,
\[
	t(e) + \pi'(w) - \pi'(v) = t(e) + \pi(w) - d(w) - \pi(v) + d(v) = 0\quad.
\]
This completes the proof.
\qed
\end{proof}

Once we augment $f$, we need to update the shortest path tree. There are three possible updates:
\emph{(i)}
adding a flow may result in a new backward edge in the residual graph,
\emph{(ii)}
reducing a flow may result in a removing a backward edge in the residual graph, and
\emph{(iii)}
deleting a source from $S$ requires that the tree is updated.

In order to update the tree we will use an algorithm by~\citet{ramalingam:96:dynamic} to which we
will refer as \algrr. The pseudo-code for the modified solver is given in Algorithm~\ref{alg:fast}.

\begin{algorithm}[h]
\caption{A modified algorithm for \uncircprb.}
\label{alg:fast}
	$\Delta \define \min(\max e(v), \max -e(v))$\;
	\While {there is excess} {
		contract any edges with $f(e) \geq 3n'\Delta$\;
		$S \define \set{v \in V \mid e(v) \geq \alpha \Delta}$\;
		$Q \define \set{v \in V \mid e(v) \leq -\alpha \Delta}$\;
		$T \define $ shortest path tree from $S$ in residual graph, weighted by $t(e) + \pi(w) - \pi(v)$\;
		update $\pi$ using $T$, see Lemma~\ref{lem:dualupdate}\;
		\While {$S \neq \emptyset$ \AND $Q \neq \emptyset$} {
			select $r \in Q$\;
			$P \define $ path in $T$ from $r$ to some $s \in S$\;
			augment flow by $\Delta$ along $P$\;
			update residual graph\;
			\lIf {$e(s) < \alpha \Delta$} {delete $s$ from $S$}
			\lIf {$e(r) > -\alpha \Delta$} {delete $r$ from $Q$}
			update $T$ using~\citep{ramalingam:96:dynamic}\;
			update $\pi$ using $T$, see Lemma~\ref{lem:dualupdate}\;
		}
		$\Delta \define \Delta / 2$\;
	}
\end{algorithm}

Before going further, we need to address one technical issue. \algrr requires
that edge weights are positive, whereas we can have weights equal to 0.  We
solve this issue by adding $\epsilon = 1/n'$ to each edge. Since the original
weights are integers and a single path may have $n' - 1$ edges, at most, the
obtained shortest path tree is a valid shortest path tree for the original
weights.  We use $\epsilon$ only for computing and updating the tree; we will
not use it when we update the dual variables.

In order to update the tree, first note that the deleting the source $s$ from
$S$ is essentially the same as deleting an edge: computing the tree using $S$
as roots is equivalent to having one auxiliary root, say $\sigma$, with only
edges connecting $\sigma$ to $S$. Removing $s$ from $S$ is then equivalent to
deleting an edge $(\sigma, s)$.

The update is done by first adding the new edges, and then deleting the necessary
edges.  We first note that the edge additions do not require any updates by
\algrr.  This is because the internal structure of \algrr is a subgraph of
\emph{all} edges that can be used to form the shortest path.  Any edge that is
added will be from a child to a parent, implying that it cannot participate in
a shortest path.\!\footnote{The edge can participate later when we delete
edges.}

\begin{proposition}
\label{prop:modifytime}
Algorithm~\ref{alg:fast} runs in $\bigO{m(\min(kn, m) + n \log n)}$ time,
assuming $G$ is not weighted.
\end{proposition}

To prove the result we need the following lemma.

\begin{lemma}
\label{lem:dualbound}
At any point of the algorithm, the dual variables $\pi$ satisfy
$\pi(v) - \pi(u) \leq k$ for any $u, v$.
\end{lemma}

\begin{proof}
Let us first prove that this result holds if we have done no edge contractions.
Let $\alpha$ and $\omega$ be the vertices in $H$, enforcing the cardinality constraint. 
Assume that $v$ and $w$ are both in $W_1$.
Then Eq.~\ref{eq:dualcond} guarantees that
$\pi(u) \geq \pi(\alpha) $ and $\pi(\omega) \geq \pi(v)$ implying
\[
	\pi(v) - \pi(u) \leq \pi(\omega) - \pi(\alpha)  \leq t(\omega, \alpha) = k - 1\quad.
\]
Now assume that $v$ (and/or $u$) is in $W_2$.
Then the shortest path tree connects it to a vertex $x \in W_1$, and either
$\pi(u) = \pi(x)$ or $\pi(u) = \pi(x) - 1$. This leads to that the difference
$\pi(v) - \pi(u)$ can be at most $k$.

To see why the lemma holds despite edge contractions, note that we can always
unroll the contractions to the original graph, and obtain $\pi'$ that satisfies
Eq.~\ref{eq:dualcond}. Moreover, if $x$ is a new vertex resulted from a
contraction, after unrolling, there is a vertex $u \in W$ such
that $\pi(x) = \pi'(u)$.  This is because when we create $x$, we initialize
$\pi(x)$ to be dual of one contracted vertices.  Consequently, the general
case reduces to the first case.\qed
\end{proof}

\begin{proof}[of Proposition~\ref{prop:modifytime}]
Since $b(v) = -1$, for $v \in W_2$ and $b(v) \geq 0$ for $v \in W_1$,
we have $\Delta = 1$, and after a single iteration $e(v) = 0$, for $v \in W$.
So, we need only one outer iteration.
Consequently, we only need to show that the inner loop needs $\bigO{m(\min(kn, m) + n \log n)}$
time.

Let us write $O$ to be the vertices who are either in $W_1$, or, due to a
contraction, contain a vertex in $W_1$.
Let $P_i$ be the path selected during the $i$th iteration.
Let us write $n_i'$ to be the number of vertices whose distance is
changed\footnote{taking into account the $\epsilon$ trick} during the $i$th
iteration of the inner loop; let $m_i$ be the number of edges adjacent to these vertices.
Finally, let us write $n_i$ to be the number of vertices in $O$ whose distance is
changed.

\citet{ramalingam:96:dynamic} showed that updating a tree during the $i$th iteration requires
$\bigO{m_i + n_i' \log n_i'}$ time. More specifically, the update algorithm
first detects the affected vertices in $\bigO{m_i}$ time, and then computes the
new distances using a standard Dijkstra algorithm with a binomial heap in
$\bigO{m_i + n_i' \log n_i'}$ time.

We can optimize this to $\bigO{m_i + n_i' + n_i \log n_i}$ by performing a trick
suggested by Orlin: Let $X$ be the vertices counted towards $n_i$ and let $Y$
be the remaining vertices counted towards $n_i'$.  A vertex in $Y$ is either in
a path between two vertices in $X$, or is a leaf. In the latter case it may be
only connected to only two (known) vertices. We can first compute the distances
for $X$ by frog-leaping the vertices in $Y$ in $\bigO{m_i + n_i \log n_i}$ time.
This gives us the updated distances for $X$ and for vertices in $Y$ that are
part of some path. Then we can proceed to update the leaf vertices in $Y$ in
$\bigO{m_i + n_i}$ time.

The total running time of an inner loop is then
\[
	\bigO{\sum_i \abs{P_i} + m_i + n_i' + n_i \log n_i} \subseteq
\bigO{\sum_i \abs{P_i} + m_i + n_i' + n_i \log n}\quad.
\]
First note that we can have at most $\bigO{m}$ terms in the sum. This is because
we either have
$\sum \max(e(i), 0) \leq 2\Delta \alpha n'$ or $\sum \max(-e(i), 0) \leq 2\Delta \alpha n'$
due to to the previous outer loop iteration, and since the contractions can only reduce these terms.

A path from a leaf to a root in $T$ cannot contain two consecutive vertices
that are outside $O$.  Hence, the length of a path is at most $\bigO{n}$.

Let us now bound the number of times a single vertex, say $v$, needs to be
updated.  Assume that we have changed the distance but the dual $\pi(v)$ has
not changed. In other words, we have increased the $\epsilon$ part of the
distance. This effectively means that $\pi(v)$ remained constant but we have
increased the number of edges from the vertex to the root. Since we can have at
most $\bigO{n}$ long path, we can have at most $\bigO{n}$ updates without
before updating $\pi(v)$. Note that at least one root, say $s \in S$, will not
have its dual updated until the very last iteration. Lemma~\ref{lem:dualbound}
now implies that we can update, that is, decrease, $\pi(v)$ only $\bigO{k}$ times.
Consequently, we can only update $v$ $\bigO{nk}$ times.

This immediately implies that $\sum_i n_i'\ \in \bigO{mnk}$, $\sum_i n_i\ \in \bigO{n^2k}$,
and $\sum_i m_i\ \in \bigO{mnk}$. Since path lengths are $\bigO{n}$, we also have $\sum_i \abs{P_i} \in \bigO{mn}$.
This gives us a total running time of
\[
	\bigO{mn + mnk + n^2k \log n}  = \bigO{nk(m + n \log n)}\quad.
\]
We obtain 
the final bound by alternatively bounding \algrr with $\bigO{m + n \log n}$,
and observing that you need only $\bigO{m}$ updates.\qed
\end{proof}

The theoretical improvement is modest: we essentially replaced $m$ with
$\min(nk, m)$.  However, in practice this is a very pessimistic bound, and we
will see that this approach provides a significant speed-up. Moreover, this
result suggests---backed up by our experiments---that the problem is easier for
smaller values of $k$. This is opposite to the behavior of the original solver
presented in~\citep{tatti:15:agony}. Note also, that we assumed that $G$ has
no weights. If we have integral weights of at most $\ell$, then the running
time increases by $\bigO{\log \ell}$ time.

\section{Alternative penalty functions}\label{sec:score}

We have shown that we can find ranking minimizing edge penalties $\penlin{}$ in
polynomial time. In this section we consider alternative penalties.
More specifically, we consider convex penalties which are solvable in polynomial time,
and show that concave penalties are \np-hard.

\subsection{Convex penalty function}\label{sec:convex}

We say that the penalty function is \emph{convex} if $\pen{x} \leq (\pen{x - 1} +
\pen{x + 1}) / 2$ for every $x \in \integers$.

Let us consider a penalty function that can be written as
\[
	\pensum{x} = \sum_{i = 1}^{\ell} \max (0, \,\alpha_i (x - \beta_i)),
\]
where $\alpha_i > 0$ and $\beta_i \in \integers$ for $1 \leq i \leq \ell$.
This penalty function is convex. On the other hand, if we are given a convex
penalty function $\pen{}$ such that $\pen{x} = 0$ for $x < 0$, then we can
safely assume that an optimal rank assignment will have values between $0$ and $\abs{V} - 1$.
We can define a penalty function $\pensum{}$ 
with $\ell \leq \abs{V}$ terms such that $\pensum{x} = \pen{x}$ for $x < \abs{V}$.
Consequently, finding an optimal rank assignment using $\pensum{}$ will also yield an optimal
rank assignment with respect to $\pen{}$.

Note that $\penlin{}$ is a special case of $\pensum{}$. This hints that we can
solve $\score{G, k, \pensum{}}$ with a technique similar to the one given in Section~\ref{sec:algo}.
In fact, we can map this problem to \genagonyprb.
In order to do this, assume a graph $G = (V, E, w)$ and an integer $k$.
Set $n = \abs{V}$ and $m = \abs{E}$.
We define a graph $H = (W, F, w, s)$ as follows.
The vertex set $W$ consists of 2 groups:
\emph{(i)} $n$ vertices, each vertex corresponding to a vertex in $G$
\emph{(ii)} $2$ additional vertices $\alpha$ and $\omega$.
For each edge $e = (v, w) \in E$, we add $\ell$ edges $f_i = (u, v)$ to $F$.
We set $s(f_i) = -\beta_i$ and $w(f_i) = \alpha_i w(e)$.
We add edges to $\alpha$ and $\omega$ to enforce the cardinality constraint,
as we did in Section~\ref{sec:shift}.
We denote this graph by $\cg{G, k, \pensum{}} = H$.

\begin{example}
\begin{figure}[ht!]
\hspace*{\fill}
\begin{tikzpicture}[baseline=0]
\node[exnode] at (0, 0) (a) {$a$};
\node[exnode] at (2, 0) (b) {$b$};
\node[exnode] at (1, -1.7) (c) {$c$};
\node at( 0, -1.7) {$G$};

\draw (a) edge[exedge] node[auto = right, black, font = \scriptsize, circle, inner sep = 1pt] {1} (b);
\draw (b) edge[exedge] node[auto = right, black, font = \scriptsize, circle, inner sep = 1pt] {1} (c);
\draw (c) edge[exedge] node[auto = right, black, font = \scriptsize, circle, inner sep = 1pt] {2} (a);

\end{tikzpicture}\hfill
\begin{tikzpicture}[baseline=0]
\node[exnode, label={[font = \scriptsize, inner sep = 1pt]above:6}] at (0, 0) (a) {$a$};
\node[exnode, label={[font = \scriptsize, inner sep = 0pt]20:3}] at (2, 0) (b) {$b$};
\node[exnode, label={[font = \scriptsize, inner sep = 1pt]below:3}] at (1, -1.7) (c) {$c$};
\node at (-1, -1.7) {$\cg{G, 3, \pensum{}}$};

\draw (a) edge[exedge, bend left = 15] node[labnode, auto = left, pos = 0.45] {$1(1)$} (b);
\draw (a) edge[exedge, bend right = 5] node[labnode, auto = right, pos = 0.45] {$-3(2)$} (b);

\draw (b) edge[exedge, bend left = 15] node[labnode, auto= left, pos = 0.7] {$1(1)$} (c);
\draw (b) edge[exedge, bend right = 5] node[labnode, auto= right, pos = 0.3] {$-3(2)$} (c);

\draw (c) edge[exedge, bend left = 15] node[labnode, auto = left, pos = 0.2] {$1(2)$} (a);
\draw (c) edge[exedge, bend right = 5] node[labnode, auto = right, pos = 0.8] {-$3(4)$} (a);

\end{tikzpicture}\hspace*{\fill}
\caption{Toy graph $G$ and the related circulation graph $\cg{G, 3, \pensum{}}$. 
To avoid clutter the vertices $\alpha$ and $\omega$ and the adjacent edges are omitted.}
\label{fig:toymulti}
\end{figure}
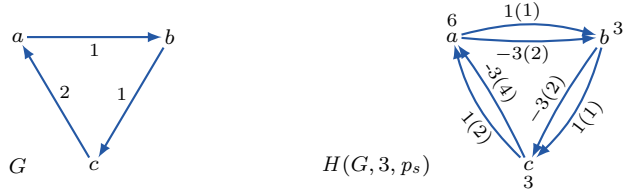

Consider a graph $G$ given in Figure~\ref{fig:toymulti} and a penalty function
$\pensum{d} = \max(0, d + 1) + 2\max(0, d - 3)$.  The graph $H = \cg{G, 3,
\pensum{}}$ has $5$ vertices, the original vertices and the two additional 
vertices. Each edge in $G$ results in two edges in $H$. This gives us 6 edges plus
the 7 edges adjacent to $\alpha$ or $\omega$.
The graph $H$ without $\alpha$ and $\omega$ is given in Figure~\ref{fig:toymulti}.
\end{example}

Finally, let us address the computational complexity of the problem.  The
circulation graph $\cg{G, k, \pensum{}}$ will have $n + 2$ vertices and $\ell m
+ n$ edges.  If the penalty function $\pen{}$ is convex, then we need at most $\ell = n$
functions to represent $\pen{}$ between the range of $[0, n - 1]$.
Moreover, if we enforce the cardinality constraint $k$, we need only $\ell = k$ 
components. Consequently, we will have at most $d m + n$, edges where
$d = \min (k, \ell, n)$ for $\pensum{}$, and 
$d = \min (k, n)$ for a convex penalty $\pen{}$. This gives us computational time
of $\bigO{d m \log n (d m + n \log n)}$.

\subsection{Concave penalty function}

We have shown that we can solve Problem~\ref{prb:opt} for any convex penalty.
Let us consider concave penalties, that is penalties for which
$\pen{x} \geq (\pen{x - 1} + \pen{x + 1}) / 2$.
There is a stark difference compared to the convex penalties as the
minimization problem becomes computationally intractable.

\begin{proposition}
\label{prop:concave}
Assume a monotonic penalty function $\funcdef{\pen{}}{\integers}{\reals}$ such
that $\pen{x} = 0$ for $x < 0$, $\pen{2} > \pen{1}$, and there is an integer $t$ such
that 
\begin{equation}
\label{eq:ineq}
	\pen{t} > \frac{\pen{t - 1} + \pen{t + 1}}{2}
\end{equation}
and
\[
	\frac{\pen{s}}{s + 1} \geq \frac{\pen{y}}{y + 1},
\]
for every $0 \leq s \leq y$ and $y \in [t - 1, t, t + 1]$.
Then, determining whether $\score{G, k, p} \leq \sigma$
for a given graph $G$, integer $k$, and threshold $\sigma$ is an \np-hard problem.
\end{proposition}

We provide the proof in Appendix.

While the conditions in Proposition~\ref{prop:concave} seem overly complicated,
they are quite easy to satisfy. Assume that we are given a penalty function
that is concave in $[-1, \infty]$, and $\pen{-1} = 0$. Then due to concavity we have
\[
	\frac{\pen{x}}{x + 1} \geq \frac{\pen{x + 1}}{x + 2},\quad\text{for}\quad x \geq 0\quad.
\]
This leads to the following corollary.
\begin{corollary}
\label{cor:concave}
Assume a monotonic penalty function $\funcdef{\pen{}}{\integers}{\reals}$ such
that $\pen{x} = 0$ for $x < 0$, $\pen{2} > \pen{1}$, and $\pen{}$ is concave
and non-linear in $[-1, \ell]$ for some $\ell \geq 1$.
Then, determining whether $\score{G, k, p} \leq \sigma$
for a given graph $G$, integer $k$, and threshold $\sigma$ is \np-hard problem.
\end{corollary}

Note that we require $\pen{}$ to be non-linear. This is needed so that the
proper inequality in Equation~\ref{eq:ineq} is satisfied. This condition is
needed since $\penlin{}$ satisfies every other requirement.
Corollary~\ref{cor:concave} covers many penalty functions such as
$\pen{x} = \sqrt{x + 1}$ or $\pen{x} = \log(x + 2)$, for $ x \geq 0$.
Note that the function needs to be convex only in $[-1, \ell]$ for some $\ell \geq 1$.
At extreme, $\ell = 1$ in which case $t = 0$ satisfies the conditions in Proposition~\ref{prop:concave}.

\section{Selecting canonical solution}\label{sec:canon}

A rank assignment minimizing agony may not be unique. In fact, consider a graph
$G$ with no edges, then any ranking will have the optimal score of $0$.
Moreover, if the input graph $G$ is a DAG, then any topological sorting of
vertices will yield the optimal score of $0$.

In this section we introduce a technique to select a unique optimal solution.
The idea here is to make the ranks as small as possible without compromising
the optimality of the solution. More specifically, let us define the following
relationship between to rankings.

\begin{definition}
Given two rank assignments $r$ and $r'$, we write $r \preceq r'$
if $r(x) \leq r'(x)$ for every $x$.
\end{definition}

The following proposition states that there exists exactly one ranking with
the optimal score that is minimal with respect to the $\preceq$ relation.
We will refer to this ranking as \emph{canonical ranking}.

\begin{proposition}
\label{prop:unique}
Given a graph $G$ and an integer $k$,
there exists a unique optimal rank assignment $r$ such that
$r \preceq r'$ for every optimal rank assignment $r'$.
\end{proposition}

The proof of this proposition is given in Appendix.

Canonical ranking has many nice properties. The canonical solution for a graph
without edges assigns rank $0$ to all vertices. More generally, if $G =
(V, E)$ is a DAG, then the source vertices $S$ of $G$ will receive a rank of
$0$, the source vertices of $G(V \setminus S)$ will receive a rank of $1$, and
so on. For general graphs we have the following proposition.

\begin{proposition}
Let $r$ be the canonical ranking. Then $r$ has the least distinct rank values
among all optimal solutions.
\end{proposition}

In other words, the partition of $V$ corresponding to the canonical ranking has
the smallest number of groups.

Our next step is to provide an algorithm for discovering canonical ranking.  In
order to do so we assume that we use Orlin's algorithm and obtain the flow $f$
and the dual $\pi$, described in
Problem~\ref{prb:uncirc}~and~\ref{prb:dualuncirc}.  We construct the residual
graph $R$, as described in Section~\ref{sec:speedup}, edges weighted by $t(e) +
\pi(w) - \pi(v)$. We then compute, $d(v)$ which is the shortest path in $R$ from $\alpha$
to $v$. Finally, we set $r^*(v) = r(v) - d(v)$.

Once, we have computed the residual graph, we simply compute the shortest path
distance from $q$ and subtract the distance from the optimal ranking, see
Algorithm~\ref{alg:canon}.

\begin{algorithm}
\caption{$\canon(G)$, computes canonical optimal solution}
\label{alg:canon}
$f, \pi \define $ optimal flow and dual of \uncircprb\;
$R \define $ residual graph\;
$d(v) \define $ shortest weighted distance from $\alpha$ to $v$ in $R$\;
\lForEach{$v \in V$} {$r^*(v) \define r(v) - d(v)$}
\Return $r^*$\;
\end{algorithm}

\begin{proposition}
\label{prop:canon}
Algorithm \canon returns canonical solution with optimal score.
\end{proposition}

We give the proof of this proposition in Appendix.

Proposition states that to compute the canonical ranking it is enough to form the
residual graph, compute the shortest edge distances $d(v)$ from the vertex $q$,
and subtract them from the input ranking. The computational complexity of
these steps is $\bigO{m + n \log n}$.
Moveover, this proposition holds for a more general
convex penalty function, described in Section~\ref{sec:convex}. 

\section{A fast divide-and-conquer heuristic}
In this section we propose a simple and fast divide-and-conquer approach.  The
main idea is as follows: We begin with the full set of vertices and we split
them into two halves: the left half will have smaller ranks than the right half.
We then continue splitting the smaller sets recursively, and obtain a tree. We show that this can be
done in $\bigO{m \log n}$ time. If we are given a cardinality constraint $k$, then
we prune the tree using dynamic program that runs in $\bigO{k^2 n}$ time.
We also propose a variant, where we perform SCC decomposition, and perform
then divide-and-conquer on individual components. To enforce the cardinality constraint in
this case, we need additional $\bigO{km \log n + k^2n}$ time.

\subsection{Constructing a tree by splitting vertices}

As mentioned above, our goal is to construct a tree $T$.  This tree is binary
and ordered, that is, each non-leaf vertex has a left child and a right child.

Each leaf $\alpha$\footnote{we will systematically denote the vertices in $T$
with Greek letters} in this tree $T$ is associated with a set of vertices that
we denote by $V_\alpha$. Every vertex of the input graph should belong to some
leaf, and no two leaves share a vertex. If $\alpha$ is a non-leaf, then
we define $V_\alpha$ to be the union of vertices related to each descendant
leaf of $\alpha$. We also define $E_\alpha$ to be the edges in $E$ that have both
endpoints in $V_\alpha$.

Since the tree is ordered, we can sort the leaves, left first.  Using this
order, we define a rank $r(v)$ to be the rank of the leaf in which $v$ is
included. We define $\score{T} = \score{r}$.

Our goal is to construct $T$ with good $\score{T}$.  We do this by splitting
$V_\alpha$ of a leaf $\alpha$ to two leaves such that the agony is minimized.

Luckily, we can find the optimal split efficiently. Let us first express the
gain in agony due to a split. In order to do so, assume a tree $T$, and let $\alpha$ be a
leaf. Let $X$ be the vertices in leaves that are left to $\alpha$, and let $Z$
be the vertices in leaves that are right to $\alpha$.

We define $\back{\alpha}$ to be the total weight of the edges from $Z$ to $X$,
\[
	\back{\alpha} = \sum_{(z, x) \in E \atop x \in X, z \in Z} w(z, x)\quad.
\]
Let $y$ be a vertex in $V_\alpha$. We define
\[
	\inback{y; \alpha} = \sum_{(z, y) \in E \atop z \in Z} w(z, y) \quad\text{and}\quad 
	\outback{y; \alpha} = \sum_{(y, x) \in E \atop x \in X} w(y, x) 
\]
to be the total weight of the backward edges adjacent to $y$ and $Z$ or $X$.
We also define the total weights
\[
	\inback{\alpha} = \sum_{y \in V_\alpha} \inback{y; \alpha} \quad\text{and}\quad 
	\outback{\alpha} = \sum_{y \in V_\alpha} \outback{y; \alpha} \quad.
\]
Let
\[
	\flux{y; \alpha} = \sum_{(x, y) \in E_\alpha} w(x, y) - \sum_{(y, x) \in E_\alpha} w(y, x)
\]
to be the total weight of incoming edges minus the total weight of the outgoing edges.

Finally, let us define
\[
	\diff{y; \alpha} = \flux{y; \alpha} + \inback{y; \alpha} - \outback{y; \alpha}\quad.
\]

We can now use these quantities to express how a split changes the score.

\begin{proposition}
\label{prop:split}
Let $\alpha$ be a leaf of a tree $T$. 
Assume a new tree $T'$, where we have split $\alpha$ to two leaves. Let $Y_1$ be the
vertex set of the new left leaf, and $Y_2$ the vertex set of the new right leaf. Then
the score difference is
\[
	\score{T'} - \score{T} = \back{\alpha} + \inback{\alpha} - \sum_{y \in Y_2} \diff{y; \alpha} 
\]
that can be rewritten as
\[
	\score{T'} - \score{T} = \back{\alpha} + \outback{\alpha} + \sum_{y \in Y_1} \diff{y; \alpha}\quad. 
\]
\end{proposition}

\begin{proof}

We will show that
\begin{equation}
\label{eq:scorediff}
	\score{T'} - \score{T} = \back{\alpha}
	+ \sum_{y \in Y_1} \inback{y; \alpha}
	+ \sum_{y \in Y_2} \outback{y; \alpha}
	+ \sum_{y \in Y_1} \flux{y; \alpha}\quad. 
\end{equation}
Equation~\ref{eq:scorediff} can be then rewritten to the forms given in the proposition. 

Let $Y_0$ be the set of all vertices to the left of $\alpha$.
Let $Y_3$ be the set of all vertices to the right of $\alpha$.
Note that $Y_0 \cup Y_1 \cup Y_2 \cup Y_3 = V$.
Write $t(i, j)$ to be the total weight of edges from $Y_i$ to $Y_j$.
Also, write $c(i, j)$ to be the total change in the penalty of edges from $Y_i$ to $Y_j$
due to a split.

Note that $c(0, 1) = c(0, 2) = c(0, 3) = c(1, 3) = c(2, 3) = 0$ since these are 
forward edges that remain forward. Also, $c(0, 0) = c(1, 1) = c(2, 2) = c(3, 3) = 0$
since the rank difference of these edges has not changed.
For the same reason, $c(3, 2) = c(1, 0) = 0$.

Case (\emph{i}):
Since a split shifts $Y_3$ by one rank, $c(3, 0) = t(3, 0)$
and $c(3, 1) = t(3, 1)$.
Case (\emph{ii}):
Since a split shifts $Y_2$ by one rank, $c(2, 0) = t(2, 0)$.
Case (\emph{iii}):
The penalty of an edge from $Y_2$ to $Y_1$ increases by $w(e)$.
Summing over these edges leads to $c(2, 1) = t(2, 1)$.
Case (\emph{iv}):
The penalty of an edge from $Y_1$ to $Y_2$ decreases by $w(e)$.
Summing over these edges leads to $c(1, 2) = -t(1, 2)$.

This leads to
\[
	\score{T'} - \score{T} = \sum_{i, j} c(i, j) = t(3, 0) + t(3, 1) + t(2, 0) + t(2, 1) - t(1, 2)\quad.
\]
First, note that 
\[
	t(3, 0) = \back{\alpha}, \quad
	t(3, 1) = \sum_{y \in Y_1} \inback{y; \alpha}, \quad
	t(2, 0) = \sum_{y \in Y_2} \outback{y; \alpha} \quad.
\]

To express $t(2, 1) - t(1, 2)$, we can write
\[
\begin{split}
	\sum_{y \in Y_1} \flux{y; \alpha} & = \sum_{(x, y) \in E_\alpha \atop y \in Y_1} w(x, y) - \sum_{(y, x) \in E_\alpha \atop y \in Y_1} w(y, x) \\
	&  = t(1, 1) + t(2, 1) -  t(1, 1) - t(1, 2) = t(2, 1) - t(1, 2) \quad.
\end{split}
\]
This proves Eq.~\ref{eq:scorediff}, and the proposition.\qed
\end{proof}

Proposition~\ref{prop:split} gives us a very simple algorithm for finding an
optimal split: A vertex $y$ for which $\diff{y} \geq 0$ should be in the
right child, while the rest vertices should be in the left child.
If the gain is negative, then we have improved the score by splitting.
However, it is possible to have positive gain, in which case we should not do a split at all.
Note that the gain does not change if we do a split in a different leaf. This allows
to treat each leaf independently, and not care about the order in which leaves are tested.

The difficulty with this approach is that if we simply recompute the quantities
every time from the scratch, we cannot guarantee a fast computation time. This
is because if there are many uneven splits, we will enumerate over some edges
too many times. In order to make the algorithm provably fast, we argue that we
can  detect which of the new leaves has \emph{fewer} adjacent edges, and we only enumerate
over these edges.

Let us describe the algorithm in more details.  We start with
the full graph, but as we split the vertices among leaves, we only keep the
edges that are intra-leaf; we delete any cross-edges between different leaves.
As we delete edges, we also maintain 4 counters for each vertex, $\flux{y; \alpha}$,
$\inback{y; \alpha}$, $\outback{y; \alpha}$, and the unweighted degree, $\deg(y)$, where $\alpha$ is the leaf containing $y$. 

For each leaf $\alpha$, we maintain four sets of vertices, 
\begin{eqnarray*}
	N_\alpha   & = & \set {y \in V_\alpha \mid \deg(y; \alpha) > 0, \diff{y; \alpha} < 0}, \\
	P_\alpha   & = & \set {y \in V_\alpha \mid \deg(y; \alpha) > 0, \diff{y; \alpha} \geq 0}, \\
	N_\alpha^* & = & \set {y \in V_\alpha \mid \deg(y; \alpha) = 0, \diff{y; \alpha} < 0}, \\
	P_\alpha^* & = & \set {y \in V_\alpha \mid \deg(y; \alpha) = 0, \diff{y; \alpha} \geq 0}\quad.
\end{eqnarray*}
The reason why we treat vertices with zero degree differently is so that we can
bound $\abs{N_\alpha}$ or $\abs{P_\alpha}$ by the number of adjacent edges.

Note that we maintain these sets only for leaves.  To save computational time,
when a leaf is split, its sets are reused by the new leaves, and in the process
are modified. 

In addition, we maintain the following counters
\begin{enumerate}
\item the total weights $\back{\alpha}$, $\inback{\alpha}$, $\outback{\alpha}$, and
\item in order to avoid enumerating over $N_\alpha^*$ and $P_\alpha^*$ when computing the gain, we also maintain the counters
\[
	\dbackone{\alpha} = \sum_{y \in N_\alpha} \inback{y} - \outback{y}, \quad \dbacktwo{\alpha} = \sum_{y \in P_\alpha} \inback{y} - \outback{y}\quad.
\]
\end{enumerate}

We also maintain $\gain{\alpha}$ for non-leaves, which is the agony gain of splitting $\alpha$.
We will use this quantity when we prune the tree to enforce the cardinality constraint.

If we decide to split, then we can do this trivially: according to
Proposition~\ref{prop:split} $N_\alpha$ and $N_\alpha^*$ should be in the left child while
$P_\alpha$ and $P_\alpha^*$ should be in the right child. Our task is to compute the
gain, and see whether we should split the leaf, and compute the structures for the new leaves.

Given a leaf $\alpha$, our first step is to determine whether $N_\alpha$ or $P_\alpha$ has fewer edges.
More formally, we define $\adj{X}$ to be the edges that have \emph{at least one}
end point in $X$. We then need to compute whether $\abs{\adj{N_\alpha}} \leq \abs{\adj{P_\alpha}}$.
This is done by cleverly enumerating over elements of $N_\alpha$ and $P_\alpha$ simultaneously.
The pseudo-code is given in Algorithm~\ref{alg:leftsmaller}.

\begin{algorithm}[ht!]
\caption{$\algcount(\alpha)$, tests whether $\abs{\adj{N_\alpha}} \leq \abs{\adj{P_\alpha}}$.}
\label{alg:leftsmaller}
	$Y_1 \define N_\alpha, Y_2 \define P_\alpha$\; 
	$c_1 \define 0$; $c_2 \define 0$\;

	\Until {($Y_1 = \emptyset$ \AND $c_1 \leq c_2$) \OR ($Y_2 = \emptyset$ \AND $c_1 \geq c_2$)} {
		\eIf {$c_1 \leq c_2$} {
			$y \define $ vertex in $Y_1$; delete $y$ from $Y_1$\;
			$c_1 \define c_1 + \deg(y)$\;
		}{
			$y \define $ vertex in $Y_2$; delete $y$ from $Y_2$\;
			$c_2 \define c_2 + \deg(y)$\;
		}
	}
	\Return $Y_1 = \emptyset$ \AND $c_1 \leq c_2$\;
\end{algorithm}

\begin{proposition}
Let $m_1 = \abs{\adj{N_\alpha}}$ and $m_2 = \abs{\adj{P_\alpha}}$. 
Then \algcount returns true if and only if $m_1 \leq m_2$ in $\bigO{\min (m_1, m_2)}$ time. 
\end{proposition}

\begin{proof}
Assume that the algorithm returns true, so $Y_1 = \emptyset$ and $c_1 \leq c_2$.
Since $Y_1 = \emptyset$, then $c_1 = m_1$, which leads to $m_1 = c_1 \leq c_2 \leq m_2$. 
Assume that the algorithm returns false.
Then the while loop condition guarantees that $Y_2 = \emptyset$ and $c_1 \geq c_2$.
Since $Y_2 = \emptyset$, then $c_2 = m_2$. Either $Y_1 \neq \emptyset$ or $c_1 > c_2$.
If latter, then $m_2 = c_2 < c_1 \leq m_1$.  If former, then $m_2 = c_2 \leq c_1 < m_1$.
This proves the correctness.

To prove the running time, first note, since there are no singletons, each
iterations will increase either $c_1$ or $c_2$. Assume that $m_1 \leq m_2$.  If we
have not terminated after $2m_1$ iterations, then we must have $m_1 < c_2$.
Since $c_1 \leq m_1 < c_2$, we will then only increase $c_1$. This requires at
most $m_1$ iterations (actually, we can show that we only need 1 more
iteration). In conclusion, the algorithm runs in $\bigO{m_1}$ time. The case for $m_1 \geq m_2$ is similar.
\qed
\end{proof}

We can now describe our main algorithm, given in Algorithms~\ref{alg:split},~\ref{alg:left},~and~\ref{alg:right}.
\algsplit is given a leaf $\alpha$. As a first step, \algsplit determines which side has fewer edges using
\algcount. After that it computes the gain, and checks whether a split is profitable. If it is, then it
calls either \algleft or \algright, depending which one is faster. These two algorithms perform the actual split
and updating the structures, and then recurse on the new leaves.

\begin{algorithm}[ht!]
\caption{$\algsplit(\alpha)$, checks if we can improve by splitting $\alpha$, and decides which side is more economical to split. Calls
either \algleft or \algright to update the structures.}
\label{alg:split}
	\eIf {$\algcount(\alpha)$} {
		$g \define \back{\alpha} + \outback{\alpha} + \dbackone{\alpha} + \sum_{y \in N_\alpha} \diff{y; \alpha}$\; 
		\lIf{$g < 0$} {
			$\algleft(\alpha)$;
			$\gain{\alpha} \define g$
		}
	}{
		$g \define \back{\alpha} + \inback{\alpha} - \dbacktwo{\alpha} - \sum_{y \in P_\alpha} \diff{y; \alpha}$\; 
		\lIf{$g < 0$} {
			$\algright(\alpha)$;
			$\gain{\alpha} \define g$
		}
	}
\end{algorithm}

\begin{algorithm}[ht!]
\caption{$\algleft(\alpha)$, performs a single split using $N_\alpha$. Recurses to \algsplit for further splits.}
\label{alg:left}

	create a new leaf $\beta$ with sets $N_\beta = N_\alpha$, $P_\beta = \emptyset$, $N_\alpha^* = P_\alpha^*$, and $P_\alpha^* = \emptyset$\;
	$\back{\beta} \define \back{\alpha} + \outback{\alpha}$\;

	create a new leaf $\gamma$ with sets $N_\gamma = \emptyset$, $P_\gamma = P_\gamma$, $N_\gamma^* = \emptyset$, and $P_\gamma^* = P_\gamma^*$\;

	$\back{\gamma} \define \back{\alpha}$\;

	\ForEach {$x \in N_\alpha$} {
		$\back{\gamma} \define \back{\gamma} + \inback{x}$\;
		$\back{\beta} \define \back{\beta} - \outback{x}$\;

		delete edges $(x, z)$ or $(z, x)$ for any $z \in P_\alpha$, and update $\flux{}$, $\deg{}$, $\inback{}$, $\outback{}$ \;
	}
	check the affected vertices and update $P_\beta$, $N_\beta$, $P_\beta^*$, $N_\beta^*$, $\dbackone{\beta}$ and $\dbacktwo{\beta}$\;
	check the affected vertices and update $P_\gamma$, $N_\gamma$, $P_\gamma^*$, $N_\gamma^*$, $\dbackone{\gamma}$ and $\dbacktwo{\gamma}$\;

	$\algsplit(\beta)$; $\algsplit(\gamma)$\;
\end{algorithm}

\begin{algorithm}[ht!]
\caption{$\algright(\alpha)$, performs a single split using $P_\alpha$. Recurses to \algsplit for further splits.}
\label{alg:right}

	create a new leaf $\beta$ with sets $N_\beta = N_\alpha$, $P_\beta = \emptyset$, $N_\alpha^* = P_\alpha^*$, and $P_\alpha^* = \emptyset$\;
	$\back{\beta} \define \back{\alpha}$\; 

	create a new leaf $\gamma$ with sets $N_\gamma = \emptyset$, $P_\gamma = P_\gamma$, $N_\gamma^* = \emptyset$, and $P_\gamma^* = P_\gamma^*$\;

	$\back{\gamma} \define \back{\alpha} + \inback{\alpha}$\;

	\ForEach {$x \in P_\alpha$} {
		$\back{\gamma} \define \back{\gamma} - \inback{x}$\;
		$\back{\beta} \define \back{\beta} + \outback{x}$\;

		delete edges $(x, z)$ or $(z, x)$ for any $z \in P_\alpha$, and update $\flux{}$, $\deg{}$, $\inback{}$, $\outback{}$ \;
	}
	check the affected vertices and update $P_\beta$, $N_\beta$, $P_\beta^*$, $N_\beta^*$, $\dbackone{\beta}$ and $\dbacktwo{\beta}$\;
	check the affected vertices and update $P_\gamma$, $N_\gamma$, $P_\gamma^*$, $N_\gamma^*$, $\dbackone{\gamma}$ and $\dbacktwo{\gamma}$\;

	$\algsplit(\beta)$; $\algsplit(\gamma)$\;
\end{algorithm}

Let us next establish the correctness of the algorithm. We only need to show
that during the split the necessary structures are maintained properly.
We only show it for \algleft, as the argument is exactly the same for \algright.

\begin{proposition}
\algleft maintains the counters and the vertex sets.
\end{proposition}

\begin{proof}
During a split, our main task is to remove the cross edges between $N_\alpha$ and
$P_\alpha$ and make sure that all the counters and the vertex sets in the new leaves
are correct.

Let $y \in V_\beta$.
If there is no cross edge attached to a vertex $y$ in $E_\alpha$, then $d(y; \beta) = d(y; \alpha)$ and $\deg(y; \beta) = \deg(y; \alpha)$.
This means that we only need to check vertices that are adjacent to a cross edge, and possibly move them to a different
set, depending on $\deg(y; \beta)$ and $d(y; \beta)$.
This is exactly what the algorithm does. The case for $y \in V_\gamma$ is similar.

The only non-trivial counters are $\back{\gamma}$ and $\back{\beta}$. Note that
$\back{\gamma}$ consists of $\back{\alpha}$ as well as additional edges to $N_\alpha$, namely
$\back{\gamma} = \back{\alpha} + \sum_{y \in N_\alpha} \inback{y}$, which is exactly what algorithm computes.
Also, $\back{\beta} = \back{\alpha} + \sum_{y \in P_\alpha} \outback{y} = \back{\alpha} + \outback{\alpha} - \sum_{y \in N_\alpha} \outback{y}$.
The remaining counters are trivial to maintain as we delete edges or move vertices from one set to another.
\qed
\end{proof}

We conclude this section with the computational complexity analysis.

\begin{proposition}
Constructing the tree can be done in $\bigO{m \log n}$ time, where $m$ is the number
of edges and $n$ is the number of vertices in the input graph.
\end{proposition}

To prove the proposition, we need the following lemmas.

\begin{lemma}
\label{lem:splitleft}
Let $m = \abs{\adj{N_\alpha}}$. Updating the new leaves in  $\algleft$ can be done in $\bigO{m}$ time.
\end{lemma}

\begin{proof}
The assignments $N_\beta = N_\alpha$, $N_\beta^* = N_\alpha^*$, $P_\gamma = P_\alpha$, $P_\gamma^* = P_\alpha^*$ are
done by reference, so they can be done in constant time.  Since there are no
singletons in $N_\alpha$, there are most $2m$ vertices in $N_\alpha$. Deleting an edge is
done in constant time, so the for-loop requires $\bigO{m}$ time. There are at most
$2m$ affected vertices, thus updating the sets also can be done in $\bigO{m}$ time.
\qed
\end{proof}

\begin{lemma}
\label{lem:splitright}
Let $m = \abs{\adj{P_\alpha}}$. Updating the new leaves in  $\algright$ can be done in $\bigO{m}$ time.
\end{lemma}

The proof for the lemma is the same as the proof for Lemma~\ref{lem:splitleft}.

\begin{proof}
Let us write $m_\alpha = \min(\abs{\adj{N_\alpha}}, \abs{\adj{P_\alpha}})$ to be the
smaller of the two adjacent edges.

Lemmas~\ref{lem:splitleft}--\ref{lem:splitright} implies that the running time is $\bigO{\sum_{\alpha}
m_\alpha}$, where $\alpha$ runs over every vertex in the final tree.

We can express the sum differently:
given an edge $e$, write
\[
	i_{e\alpha} =
	\begin{cases}
		1 & e \in \adj{N_\alpha}, \quad\text{and}\quad \abs{\adj{N_\alpha}} \leq \abs{\adj{P_\alpha}}, \\
		1 & e \in \adj{P_\alpha}, \quad\text{and}\quad \abs{\adj{N_\alpha}} > \abs{\adj{P_\alpha}}, \\
		0 & \text{otherwise}\quad.
	\end{cases}
\]
That is, $m_\alpha = \sum_{e} i_{e\alpha}$.
Write $i_{e} = \sum_{\alpha} i_{e\alpha}$. To prove the proposition, we show that $i_e \in \bigO{\log n}$.

Fix $e$, and let $\alpha$ and $\beta$ be two vertices in a tree for which
$i_{e\alpha} = i_{e\beta} = 1$.  Either $\alpha$ is a descendant of $\beta$, or
$\beta$ is a descendant of $\alpha$.  Assume the latter, without the loss of
generality.
We will show that $2\abs{E_\beta} \leq \abs{E_\alpha}$, and this immediately proves that $i_e \in \bigO{\log n}$.

To prove this, let us define $c_\alpha$ to be the number of cross edges between
$N_\alpha$ and $P_\alpha$, when splitting $\alpha$.
Assume, for simplicity, that $\beta$ is the left descendant of $\alpha$.
Then $\abs{E_\beta} \leq \abs{\adj{N_\alpha}} - c_\alpha$.
Also, $e \in \adj{N_\alpha}$, and by definition of $i_{e\alpha}$, $m_\alpha = \abs{\adj{N_\alpha}}$.
This gives us,
\[
	2\abs{E_\beta} \leq
	2m_\alpha - 2c_\alpha \leq
	(\abs{\adj{N_\alpha}} - c_\alpha) + 
	(\abs{\adj{P_\alpha}} - c_\alpha) = \abs{E_\alpha} - c_\alpha \leq \abs{E_\alpha}\quad.
\]
The case when $\beta$ is the right descendant is similar, proving the result.
\qed
\end{proof}

\subsection{Enforcing the cardinality constraint by pruning the tree}
If we did not specify the cardinality constraint, then once we have obtained
the tree, we can now assign individual ranks to the leaves, and consequently to the
vertices. If $k$ is specified, then we may violate the cardinality constraint
by having too many leaves.

In such case, we need to reduce the number of leaves, which we do by pruning
some branches. Luckily, we can do this optimally by using dynamic programming.
To see this, let $T'$ be a subtree of $T$ obtained by merging some of the branches, making them into leaves.
Then Proposition~\ref{prop:split} implies that $\score{T'}$ is equal to
\[
	\score{T'} = W + \sum_{\alpha \text{ is a non-leaf in } T'} \gain{\alpha},
\]
where $W$ is the total weight of edges.

This allows us to define the following dynamic program.
Let $\opt{\alpha; h}$ be the optimal gain achieved in branch starting from $\alpha$
using only $h$ ranks. If $\alpha$ is the root of $T$, then $\opt{\alpha; k}$
is the optimal agony that can be obtained by pruning $T$ to have only $k$ leaves.

To compute $\opt{\alpha; h}$,
we first set $\opt{\alpha; 1} = 0$ for any $\alpha$, and $\opt{\alpha; h} = 0$ if $\alpha$ is a leaf in $T$. 
If $\alpha$ is a non-leaf and $k > 1$, then we need to distribute the budget among the two children,
that is, we compute
\[
	\opt{\alpha; h} = \gain{\alpha} + \min_{1 \leq \ell \leq h - 1} \opt{\beta; \ell} + \opt{\gamma; h - \ell}\quad.
\]
We also record the optimal index $\ell$, that allows us to recover the optimal tree.
Computing a single $\opt{\alpha; h}$ requires $\bigO{k}$ time, and we need to compute
at most $\bigO{nk}$ entries, leading to $\bigO{nk^2}$ running time.

\subsection{Strongly connected component decomposition}

If the input graph has no cycles and there is no cardinality constraint, then
the optimal agony is 0.  However, the heuristic is not guaranteed to produce
such a ranking. To guarantee this, we add an additional---and optional---step.
First, we perform the SCC decomposition.
Secondly, we pack strongly connected components in the minimal number of layers:
source components are in the first layer, second layer consists of components
having edges edge only from the first layer, and so on.
We then run the heuristic on each individual layer.

If $k$ is not set, we can now create a global ranking, where the ranks of the
$i$th layer are larger than the ranks of the $(i - 1)$th layer. In other words,
edges between the SCCs are all forward.

If $k$ is set, then we need to decide how many individual ranks each component
should receive. Moreover, we may have more than $k$ layers, so some of the
layers must be merged. In such a case, we will demand that the merged layers
must use exactly 1 rank, together. The reason for this restriction is that it allows
us to compute the optimal distribution quickly using dynamic programming.

The gain in agony comes from two different sources. The first source is
the improvement of edges within a single layer.
Let us adopt the notation from the previous section, and write $\opt{i; h}$ to
be the optimal gain for $i$th layer using $h$ ranks. We can compute this using
the dynamic program in the previous section.
The second source of gain is making the inter-layer edges forward. Instead of
computing the total weight of such edges, we compute how many edges are not made forward.
These are exactly the edges that are between the layers that have been merged together.
In order to express this we write $w(j, i)$ to be the total weight of inter-layer
edges having both end points in layers $j, \ldots, i$.

To express, the agony of the tree,
let $k_i$ be the budget of individual layers.
We also, write $[a_j, b_j]$ to mean that layers $a_j, \ldots, b_j$ have been
merged, and must share a single rank.
We can show that the score of the tree $T'$ that uses this budget distribution of 
is then equal to
\[
	\score{T'} = W + \sum_{i} \opt{i, k_i} + \sum_{j} w(a_j, b_j),
\]
where $W$ is the total weight of the \emph{intra-layer} edges. Note that $W$ is a constant and so we can ignore it.

To find the optimal $k_i$ and $[a_j, b_j]$, we use the following dynamic program.
Let us write $\lopt{i; h}$ to be the gain of $1, \ldots, i$ layers using $h$ ranks. 
We can express $\lopt{i; h}$ as
\[
	\lopt{i; h} = \min (\min_j w(j, i) +  \lopt{j - 1; h - 1},\  \min_\ell \opt{i, \ell} + \lopt{i - 1, h - \ell} )\quad.
\]
The first part represents merging $j, \ldots, i$ layers, while the second part represents
spending $\ell$ ranks on the $i$th layer. By recording the optimal $j$ and $\ell$ we can recover
the optimal budget distribution for each $i$ and $h$.

Computing the second part can be done in $\bigO{k}$ time, and computing the
first part can be done in $\bigO{n}$ time, naively. This leads to $\bigO{n^2k +
nk^2}$ running time, which is too expensive.

Luckily we can speed-up the computation of the first term. To simplify notation,
fix $h$, and let us write $f(j, i) = w(j, i) +  \lopt{j - 1; h - 1}$.
We wish to find $j(i)$ such that  $f(j(i), i)$ is minimal for each $i$.
Luckily, $f$ satisfies the condition,
\[
	f(j_1, i_1) - f(j_2, i_1) \leq f(j_1, i_2) - f(j_2, i_2),
\]
where $j_1 \leq j_2 \leq i_1 \leq i_2$.
\citet{aggarwal87smawk} now guarantees that $j(i_1) \leq j(i_2)$, for $i_1 \leq i_2$.
Moreover, \citet{aggarwal87smawk} provides an algorithm that computes $j(i)$ in $\bigO{n}$ time. 
Unfortunately, we cannot use it since it assumes that $f(j, i)$ can be computed
in constant time, which is not the case due to $w(j, i)$.

Fortunately, we can still use the monotonicity of $j(\cdot)$ to speed-up the
algorithm. We do this by computing $j(i)$ in an interleaved manner.
In order to do so, let $\ell$ be the number of layers, and let $t$
be the largest integer such that $s = 2^t \leq \ell$. We first compute $j(s)$.
We then proceed to compute $j(s / 2)$ and $j(3 s / 2)$, and so on.
We use the previously computed values of $j(\cdot)$ as sentinels:
when computing $j(s / 2)$ we do not test $j > j(s)$ or
when computing $j(3 s / 2)$ we do not test $j < j(s)$.
The pseudo-code is given in Algorithm~\ref{alg:fastgroup}.

\begin{algorithm}
\caption{Fast algorithm for computing $j(i)$ minimizing $f(j(i), i)$}
\label{alg:fastgroup}
	$\ell \define $ largest possible $i$\;
	$s \define \max \set{2^t \leq \ell, t \in \naturals}$\;
	\While {$s \geq 1$ } {
		\ForEach{$i = s, 3s, 5s, \ldots$, $i \leq \ell$} {
			$a \define 1$; $b \define i$\;
			\lIf {$i - s \geq 1$} {$a \define j(i - s)$}
			\lIf {$i + s \leq \ell$} {$b \define \min(b, j(i + s))$}
			$j(i) \define \min_{a \leq j \leq b} f(j, i)$\;
		}
		$s \define s / 2$\;
	}
\end{algorithm}

To analyze the complexity, note that for a fixed $s$, the variables $i$, $a$ and $b$ are only moving to the right. 
This allows us to compute $w(j, i)$ incrementally:
whenever we increase $i$, we add the weights of new edges to the total weight,
whenever we increase $j$, we delete the weights of expiring edges from the total weight.
Each edge is visited twice, and this gives us $\bigO{m}$ time for a fixed $s$.
Since $s$ is halved during each outer iteration, there can be at most $\bigO{\log n}$ iterations.
We need to do this for each $h$, so the total running time is $\bigO{km\log n + nk^2}$.

As our final remark, we should point out that using this decomposition may not
necessarily result in a better ranking.  If $k$ is not specified, then the optimal solution
will have inter-layer edges as forward, so we expect this decomposition to improve the quality.
However, if $k$ is small, we may have a better solution if we allow to inter-layer edges
go backward. At extreme, $k = 2$, we are guaranteed that the heuristic \emph{without}
the SCC decomposition will give an optimal solution, so the SCC decomposition can only harm
the solution. We will see this behaviour in the experimental section. Luckily, since both
algorithms are fast, we can simply run both approaches and select the better one.

\section{Related work}\label{sec:related}

The problem of discovering the rank of an object based on its dominating
relationships to other objects is a classic problem.
Perhaps the most known ranking method is Elo rating devised
by~\citet{elo1978rating}, used to rank chess players.  In similar fashion,
\citet{jameson:99:behaviour} introduced a statistical model, where the
likelihood of the the vertex dominating other is based on the difference of
their ranks, to animal dominance data.

\citet{DBLP:conf/cse/MaiyaB09} suggested an approach for discovering
hierarchies, directed trees from weighted graphs such that parent vertices tend
to dominate the children. To score such a hierarchy the authors propose a
statistical model where the probability of an edge is high between a parent and
a child.  To find a good hierarchy the authors employ a greedy heuristic.

The technical relationship between our approach and the previous studies on
agony by~\citet{gupte:11:agony}~and~\citet{tatti:14:agony} is a very natural
one. The authors of both papers demonstrate that minimizing agony in a
unweighted graph is a dual problem to finding a maximal eulerian subgraph, a
subgraph in which, for each vertex, the same number of outdoing edges and the
number of incoming edges is the same. Discovering the maximum eulerian subgraph
is a special case of the capacitated circulation problem, where the capacities
are set to 1. However, the algorithms in~\citep{gupte:11:agony,tatti:14:agony}
are specifically designed to work with unweighted edges. Consequently, if our
input graph edges or we wish to enforce the cardinality constraint, we need to
solve the problem using the capacitated circulation solver.

The stark difference of computational complexities for different edge penalties
is intriguing: while we can compute agony and any other convex score in
polynomial-time, minimizing the concave penalties is \np-hard. Minimizing the score
$\score{G, k, \pencons{}}$ is equivalent to \textsc{feedback arc set}
(\fasprb), which is known to be \apx-hard with a coefficient of $c =
1.3606$~\cite{dinur:05:cover}.  Moreover, there is no known constant-ratio
approximation algorithm for \fasprb, and the best known approximation algorithm
has ratio $\bigO{\log n \log \log n}$~\cite{even:98:feedback}. In this paper we
have shown that minimizing concave penalty is \np-hard. An interesting
theoretical question is whether this optimization problem is also \apx-hard,
and is it possible to develop an approximation algorithm.

Role mining, where vertices are assigned different roles based on their
adjacent edges, and other features, has received some attention.
\citet{henderson:12:roix} studied assigning roles to vertices based on its
features while \citet{mccallum:07:roles} assigned topic distributions to
individual vertices. A potential direction for a future work is to study
whether the rank obtained from minimizing agony can be applied as a feature in
role discovery. 

\section{Experiments}\label{sec:exp}

In this section we present our experimental evaluation.  Our main focus of the
experiments is practical computability of the weighted agony.

\subsection{Datasets and setup}
For our experiments we took 10 large networks from SNAP repository~\cite{snapnets}.
In addition, for illustrative purposes, we used two small datasets: \dtname{Nfl},
consisting of National Football League teams. We created an edge $(x, y)$ if
team $x$ has scored more points against team $y$ during $2014$ regular season, we assign
the weight to be the difference between the points. Since not every team plays
against every team, the graph is not a tournament. \dtname{Reef}, a food web
of guilds of species~\cite{roopnarine:13:reef}, available at~\cite{roopnarine:12:datareef}.
The dataset consisted of 3 food webs of coral reef systems: The Cayman Islands, Jamaica, and Cuba.
An edge $(x, y)$ appears if a guild $x$ is known to prey on a guild $y$. Since the guilds
are common among all 3 graphs, we combined the food webs into one graph, and weighted
the edges accordingly, that is, each edge received a weight between $1$ and $3$.

The sizes of the graphs, along with the sizes of the largest strongly connected component, are given
in the first 4 columns of Table~\ref{tab:basic}.

The 3 \dtname{Higgs} and \dtname{Nfl} graphs had weighted edges, and for the
remaining graphs we assigned a weight of 1 for each edge.  We removed any
self-loops as they have no effect on the ranking, as well as any singleton
vertices. 

For each dataset we computed the agony using Algorithm~\ref{alg:fast}.
We compared the algorithm to the baseline given by~\citet{tatti:15:agony}.
For the
unweighted graphs we also computed the agony using \relief, an algorithm
suggested by~\citet{tatti:14:agony}. Note that this algorithm, nor the
algorithm by~\citet{gupte:11:agony}, does not work for weighted graphs nor when the
cardinality constraint $k$ given in Problem~\ref{prb:opt} is enforced.
We implemented algorithms in C++ and performed experiments using a Linux-desktop
equipped with a Opteron 2220 SE processor.\!\footnote{The source code is available at \url{http://users.ics.aalto.fi/ntatti/agony.zip}}

\begin{table*}[htb!]
\caption{Basic characteristics of the datasets and the experiments. The 6th
is the number of groups in the optimal ranking.}
\label{tab:basic}
\setlength{\tabcolsep}{2.5pt}
\begin{tabular*}{\textwidth}{@{\extracolsep{\fill}}l rr rr r rrr r rrr r}
\toprule
&&&
\multicolumn{2}{l}{largest SCC} &&
\multicolumn{2}{l}{time} &
\multicolumn{2}{l}{baseline} &
\\
\cmidrule{4-5}
\cmidrule(r){7-8}
\cmidrule{9-10}
Name & $\abs{V}$ & $\abs{E}$ & $\abs{V'}$ & $\abs{E'}$ &
$k$ &
SCC & plain &
\cite{tatti:15:agony} & \cite{tatti:14:agony}
\\
\midrule
Amazon     & 403\,394 & 3\,387\,388 &  
395\,234 &
3\,301\,092
&  17
& 24m7s & 25m6s
& 6h24m & 4h27m
\\
Gnutella   &  62\,586 &    147\,892 &
14\,149 & 50\,916 
& 24 &
4s & 20s
& 8s & 45s
\\
EmailEU    & 265\,214 &    418\,956 &
 34\,203 &
 151\,132
& 9 &
10s & 29s
& 10m & 2m
\\
Epinions   &  75\,879 &    508\,837 &
 32\,223 &
 443\,506
& 10 &
33s & 44s
& 49m & 20m
\\
Slashdot   &  82\,168 &    870\,161 &
71\,307 &
841\,201
& 9 &
38s & 61s
& 1h38m & 1h5m 
\\
WebGoogle  & 875\,713 & 5\,105\,039 & 
434\,818 &
3\,419\,124
& 31 
& 10m31s & 25m22s
& 8h50m & 2h32m
\\
WikiVote   &   7\,115 &    103\,689 &
1\,300 &
39\,456
& 12 &
2s & 6s
&  43s & 7s
\\[1mm]
Nfl        &       32 &         205 &
32 &         205
& 6  
& 4ms & 5ms
& 22ms & --
\\
Reef & 258 & 4232 &
1 & 0 &
19 
& 8ms & 100ms
& 10ms &  --
\\
HiggsReply &  37\,145 &     30\,517 &
263 &
569
& 11 
& 0.3s & 5s
& 0.2s & --
\\
HiggsRetweet &
425\,008 & 733\,610 &
13\,086 &
63\,505 
& 22
& 12s & 2m10s
& 10m & --
\\
HiggsMention &
302\,523 & 445\,147  &
4\,786 &
19\,848
& 21
& 6s & 1m34s
& 2m & --
\\
\bottomrule
\end{tabular*}
\end{table*}

\subsection{Results}

Let us begin by studying running times given in Table~\ref{tab:basic}. We
report the running times of our approach with and without the strongly
connected component decomposition as suggested by Proposition~\ref{prop:decomp}, and
compare it against the baselines, whenever possible. Note that we can
use the decomposition only if we do not enforce the cardinality constraint.

Our first observation is that the decomposition always helps to speed up the
algorithm.  In fact, this speed-up may be dramatic, if the size of the strongly
connected component is significantly smaller than the size of the input graph,
for example, with \dtname{HiggsRetweet}.
The running times are practical despite the unappealing theoretical bound.
This is due to several factors. First, note that the theoretical bound 
of $\bigO{\min(nk, m)m\log n}$ given in Section~\ref{sec:speedup} only holds for unweighted graphs, and it is needed
to bound the number of outer-loop iterations. In practice, however, the number
of these iterations is small, even for weighted graphs. The other, and the main, reason
is the pessimistic $n$ in the $\min(nk, m)$ factor; we spend $nk$ inner-loop iterations only if
the dual $\pi(v)$ of each vertex $v$ increases by $\bigO{k}$, and \emph{between} the increases
the shortest path from sources to $v$ changes $\bigO{n}$ times. The latter change seems highly
unlikely in practice, leading to a faster computational time.

We see that our algorithm beats consistently both baselines. What is more
important: the running times remain practical, even if we do not use strongly
connected components.  This allows us to limit the number of groups for large
graphs. This is a significant improvement over~\citep{tatti:15:agony}, where
solving \dtname{HiggsRetweet} \emph{without} the SCC decomposition required 31 hours.

Our next step is to study the effect of the constraint $k$, the maximum number
of different rank values. We see in the 6th column in Table~\ref{tab:basic}
that despite having large number of vertices, that the optimal rank assignment
has low number of groups, typically around 10--20 groups, even if the
cardinality constraint is not enforced.

Let us now consider agony as a function of $k$, which we have plotted in Figure~\ref{fig:agonyk}
for \dtname{Gnutella} and \dtname{WikiVote} graphs. We see that for these datasets that agony
remains relatively constant as we decrease $k$, and starts to increase more prominently once we consider
assignments with $k \leq 5$.

\begin{figure}[htb!]
\hfill
\begin{tikzpicture}
\begin{axis}[xlabel={constraint $k$},ylabel= {$\score{G, k}$},
    width = 4.5cm,
	height = 3cm,
	ymin = 0,
    cycle list name=yaf,
    scale only axis,
    title = {\dtname{Gnutella}},
	tick scale binop=\times,
	xtick = {2, 5, 10, 15, 20, 24},
    no markers
    ]
\addplot table[x expr = {\coordindex + 2}, y index = 1, header = false] {gnutellastats.dat};
\pgfplotsextra{\yafdrawaxis{2}{24}{0}{50227}}
\end{axis}
\end{tikzpicture}\hfill
\begin{tikzpicture}
\begin{axis}[xlabel={constraint $k$},ylabel= {$\score{G, k}$},
    width = 4.5cm,
	height = 3cm,
	ymin = 0,
    cycle list name=yaf,
	ytick = {0, 10000, 20000, 30000, 40000, 50000},
    scale only axis,
    title = {\dtname{WikiVote}},
	tick scale binop=\times,
    no markers
    ]
\addplot table[x expr = {\coordindex + 2}, y index = 1, header = false] {wikistats.dat};
\pgfplotsextra{\yafdrawaxis{2}{12}{0}{35989}}
\end{axis}
\end{tikzpicture}\hspace*{\fill}
\caption{Agony as a function of the constraint $k$ for \dtname{Gnutella} and \dtname{WikiVote} datasets.}
\label{fig:agonyk}
\end{figure}
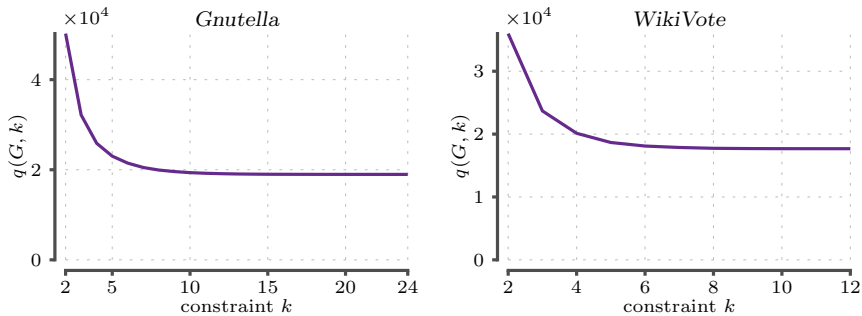

Enforcing the constraint $k$ has an impact on running time. As implied by
Proposition~\ref{prop:modifytime}, low values of $k$ should speed-up the
computation.  In Figure~\ref{fig:timek} we plotted the running time as a
function of $k$, compared to the plain version without the speed-up.

\begin{figure}[htb!]
\hfill
\begin{tikzpicture}
\begin{semilogyaxis}[xlabel={constraint $k$},ylabel= {time (in seconds)},
    width = 4cm,
    height = 3cm,
    cycle list name=yaf,
    scale only axis,
    title = {\dtname{Gnutella}},
    tick scale binop=\times,
    xtick = {2, 5, 10, 15, 20, 24},
    no markers,
    legend entries = {baseline~\citep{tatti:15:agony}, speed-up},
    minor tick length = 0pt
    ]
\addplot table[x expr = {\coordindex + 2}, y expr = {\thisrowno{0}}, header = false] {gnutellaorig.dat};
\addplot table[x index = {0},  y index = 1, header = false] {gnutellatime.dat};

\pgfplotsextra{\yafdrawaxis{2}{24}{1}{7200}}
\end{semilogyaxis}
\end{tikzpicture}\hfill%
\begin{tikzpicture}
\begin{semilogyaxis}[xlabel={constraint $k$},ylabel= {time (in seconds)},
    width = 4cm,
    height = 3cm,
    cycle list name=yaf,
    scale only axis,
    title = {\dtname{WikiVote}},
    tick scale binop=\times,
    no markers,
    xtick = {2, 4, 6, 8, 10, 12},
    minor tick length = 0pt
    ]
\addplot table[x expr = {\coordindex + 2}, y expr = {\thisrowno{0}}, header = false] {wikiorig.dat};
\addplot table[x index = {0},  y index = 1, header = false] {wikitime.dat};
\pgfplotsextra{\yafdrawaxis{2}{12}{1}{2700}}
\end{semilogyaxis}
\end{tikzpicture}\hspace*{\fill}

\caption{Execution time as a function of the constraint $k$ for \dtname{Gnutella} and \dtname{WikiVote} datasets.
Note that the $y$-axis is logarithmic.}
\label{fig:timek}
\end{figure}
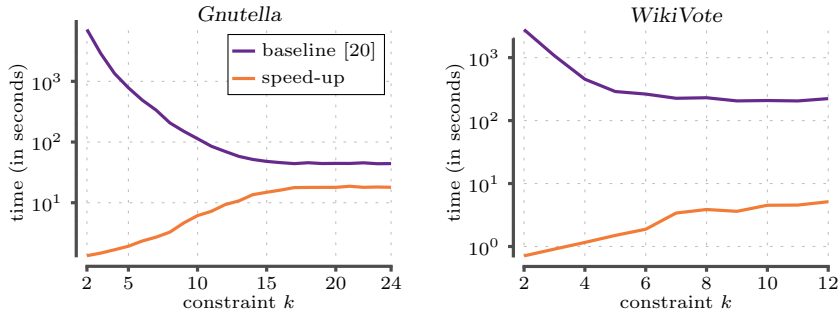

As we can see lower values of $k$ are computationally easier to solve.  This is
an opposite behavior of~\citep{tatti:15:agony}, where lowering $k$ increased
the computational time. To explain this behaviour, note that when we decrease
$k$ we increase the agony score, which is equivalent to the capacitated
circulation.  Both solvers increase incrementally the flow until we have
reached the solution. As we lower $k$, we increase the amount of optimal
circulation, and we need more iterations to reach the optimal solution. The
difference between the algorithm is that for lower $k$ updating the residual
graph becomes significantly faster than computing the tree from scratch.
This largely overpowers the effect of needing many more iterations to converge. 
However, there are exceptions: for example, computing agony for
\dtname{WikiVote} with $k = 8$ is slower than $k = 9$. 

Let us now consider the performance of the heuristic algorithm.  We report the
obtained scores and the running times in Table~\ref{tab:heuristic}.
We tested both variants: with and without SCC decomposition, and we do not enforce $k$.
We first observe that both variants are expectedly fast: processing the largest
graphs, \dtname{Amazon} and \dtname{WebGoogle}, required less than 10 seconds,
while the exact version needed $10$--$25$ minutes. The plain version is cosmetically faster.
Heuristic also produces competitive scores but the performance depends on the dataset:
for \dtname{Gnutella} and \dtname{HiggsRetweet} the SCC variant produced 25\% increase to
agony, while for the remaining datasets the increase was lower than 8\%.
Note that, \dtname{Reef} has agony of 0, that is, the network is a DAG but the plain
variant was not able to detect this. This highlights the benefit of doing the SCC decomposition. 
In general, the SCC variant outperforms the plain variant when we do not enforce the cardinality constraint.

\begin{table}
\caption{Scores, compared to the optimal, and running times of the heuristic.
Here, \emph{SCC} is the heuristic with SCC decomposition, while \emph{plain} is the plain version,
\emph{opt} is the optimal agony.
}
\label{tab:heuristic}

\begin{tabular*}{\textwidth}{@{\extracolsep{\fill}}l rrr rr rr}
\toprule
&
& & & & &
\multicolumn{2}{l}{Time (sec.)} \\
\cmidrule{7-8}
Name &
$\frac{\score{\text{SCC}}}{\score{\text{opt}}}$ &
$\frac{\score{\text{plain}}}{\score{\text{opt}}}$ &
$\score{\text{SCC}}$ &
$\score{\text{plain}}$ &
$\score{\text{opt}}$ &
SCC &
plain

\\
\midrule
Amazon     & 
1.036 &
1.037 &
2\,044\,609 &
2\,046\,344 &
1\,973\,965 &
9.24 &
8.49
\\
Gnutella   &
1.256 &
1.350 &
23\,820 &
25\,603 &
18\,964 &
0.35 &
0.34
\\

EmailEU    &
1.008 &
1.012 &
121\,820 &
122\,362 &
120\,874 &
0.47 &
0.45
\\

Epinions   & 
1.024 &
1.030 &
271\,419 &
273\,016 &
264\,995 &
0.40 &
0.37
\\

Slashdot   &
1.001 &
1.003 &
749\,448 &
750\,760 &
748\,582 &
0.68 &
0.64
\\

WebGoogle  & 
1.051 &
1.079 &
1\,935\,476 &
1\,985\,831 &
1\,841\,215 &
6.80 &
6.64
\\

WikiVote   &   
1.043 &
1.091 &
18\,430 &
19\,276 &
17\,676 &
0.05 &
0.05
\\[1mm]

Nfl        &
1.047 &
1.047 &
1172 &
1172 &
1119 &
0.002 &
0.002 
\\
Reef & 
-- &
-- &
0 &
452 &
0 &
0.008 &
0.006
\\
HiggsReply &  
1.007 &
1.103 &
5\,499 &
6\,022 &
5459 &
2.29 &
1.03
\\

HiggsRetweet &
1.259 &
1.606 &
19\,264 &
24\,579 &
15\,302 &
0.12 &
0.05
\\

HiggsMention &
1.078 &
1.322 &
24\,165 &
29\,632 &
22\,418 &
1.82 &
1.65
\\

\bottomrule
\end{tabular*}
\end{table}

As we lower the cardinality constraint $k$, the plain variant starts to
outperform the SCC variant, as shown in Figure~\ref{fig:heuristicratio}.
The reason for this is that the SCC variant requires that a edge $(u, v)$ between two SCCs
is either forward or $r(u) = r(v)$. This restriction becomes too severe as we lower $k$ and
it becomes more profitable to allow $r(u) < r(v)$. At extreme $k = 2$, the plain version
is guaranteed to find the optimal solution, so the SCC variant can only harm the solution.

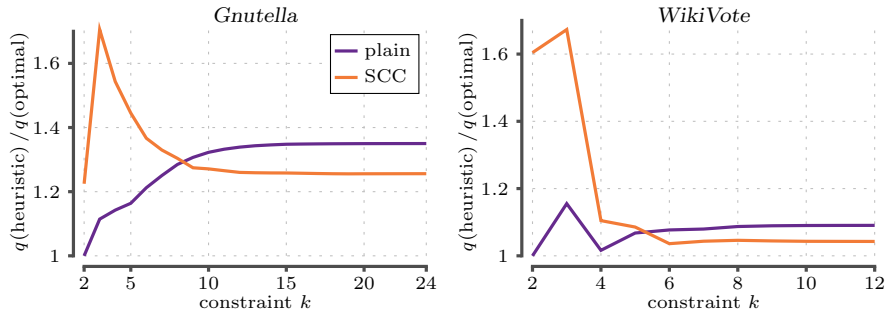
\begin{figure}[htb!]
\hfill
\begin{tikzpicture}
\begin{axis}[xlabel={constraint $k$},ylabel= {$\score{\text{heuristic}} / \score{\text{optimal}}$},
    width = 4.5cm,
	height = 3cm,
	ymin = 1,
    cycle list name=yaf,
    scale only axis,
    title = {\dtname{Gnutella}},
	tick scale binop=\times,
	xtick = {2, 5, 10, 15, 20, 24},
    legend entries = {plain, SCC},
    no markers
    ]
\addplot table[x expr = {\coordindex + 2}, y expr = {\thisrowno{5}/\thisrowno{1}}, header = false] {gnutella.stat};
\addplot table[x expr = {\coordindex + 2}, y expr = {\thisrowno{3}/\thisrowno{1}}, header = false] {gnutella.stat};
\pgfplotsextra{\yafdrawaxis{2}{24}{1}{1.7}}
\end{axis}
\end{tikzpicture}\hfill
\begin{tikzpicture}
\begin{axis}[xlabel={constraint $k$},ylabel= {$\score{\text{heuristic}} / \score{\text{optimal}}$},
    width = 4.5cm,
	height = 3cm,
	ymin = 1,
    cycle list name=yaf,
    scale only axis,
    title = {\dtname{WikiVote}},
	tick scale binop=\times,
    no markers
    ]
\addplot table[x expr = {\coordindex + 2}, y expr = {\thisrowno{5}/\thisrowno{1}}, header = false] {wiki.stat};
\addplot table[x expr = {\coordindex + 2}, y expr = {\thisrowno{3}/\thisrowno{1}}, header = false] {wiki.stat};
\pgfplotsextra{\yafdrawaxis{2}{12}{1}{1.7}}
\end{axis}
\end{tikzpicture}\hspace*{\fill}
\caption{Ratio of the agony given by the heuristic and the optimal agony as a function of the constraint $k$ for \dtname{Gnutella} and \dtname{WikiVote} datasets.}
\label{fig:heuristicratio}
\end{figure}

\begin{table}[htb!]
\caption{Rank assignment discovered for \dtname{Nfl} dataset with $k = 3$ groups}
\label{tab:nfl}
\begin{tabular*}{\columnwidth}{ll}
\toprule
Rank & Teams \\
\midrule
1. &
\textsc{den bal ne dal sea phi kc gb pit}  \\
2. & 
\textsc{stl nyg mia car no sd min cin buf det ind hou sf ari}  \\
3. &
\textsc{wsh oak tb jax ten cle atl nyj chi}  \\
\bottomrule
\end{tabular*}
\end{table}

\begin{table}
\caption{Ranked guilds of \dtname{Reef} dataset, with $k = 4$. For simplicity, we removed the duplicate guilds
in the same group, and grouped similar guilds (marked as italic, the number in parentheses indicating the number of guilds).}
\label{tab:reef}
\begin{tabular*}{\textwidth}{p{11.5cm}}
\toprule
\textsl{Sharks (6)},
Amberjack,
Barracuda,
Bigeye,
Coney grouper,
Flounder,
Frogfish,
Grouper,
Grunt,
Hind,
Lizardfish,
Mackerel,
Margate,
Palometa,
Red hind,
Red snapper,
Remora,
Scorpionfish,
Sheepshead,
Snapper,
Spotted eagle ray
\\
\midrule

Angelfish,
Atlantic spadefish,
Ballyhoo,
Barracuda,
Bass
Batfish,
Blenny
Butterflyfish,
Caribbean Reef Octopus,
Caribbean Reef Squid,
Carnivorous fish II-V,
Cornetfish,
Cowfish,
Damselfish,
Filefish,
Flamefish,
Flounder,
Goatfish,
Grunt,
Halfbeak,
Hamlet,
Hawkfish,
Hawksbill turtle,
Herring,
Hogfish,
Jack,
Jacknife fish,
Jawfish,
Loggerhead sea turtle,
Margate,
Moray,
Needlefish
Porcupinefish I-II,
Porkfish,
Pufferfish,
Scorpionfish,
Seabream,
Sergeant major,
Sharptail eel,
Slender Inshore Squid,
Slippery dick,
Snapper,
Soldierfish,
Spotted drum,
Squirrelfish,
Stomatopods II,
Triggerfish,
Trumpetfish,
Trunkfish,
Wrasse,
Yellowfin mojarra
\\
\midrule
\textsl{Crustacea (31)},
Ahermatypic benthic corals,
Ahermatypic gorgonians,
Anchovy,
Angelfish,
Benthic carnivores II,
Blenny,
Carnivorous fish I,
Common Octopus,
Corallivorous gastropods IV,
Deep infaunal soft substrate suspension feeders,
Diadema,
Echinometra,
Goby,
Green sea turtle,
Herbivorous fish I-IV,
Herbivorous gastropods I,
Hermatypic benthic carnivores I,
Hermatypic corals,
Hermatypic gorgonians,
Herring,
Infaunal hard substrate suspension feeders,
Lytechinus,
Macroplanktonic carnivores II-IV,
Macroplanktonic herbivores I,
Molluscivores I,
Omnivorous gastropod,
Parrotfish,
Pilotfish,
Silverside,
Stomatopods I,
Tripneustes,
Zooplanktivorous fish I-II,
\\
\midrule
\textsl{Planktons (7)},
\textsl{Algae (6)},
\textsl{Sponges (2)},
\textsl{Feeders (11)},
Benthic carnivores I,
Carnivorous ophiuroids,
Cleaner crustacea I,
Corallivorous polychaetes,
Detritivorous gastropods I,
Echinoid carnivores I,
Endolithic polychaetes,
Epiphyte grazer I,
Epiphytic autotrophs,
Eucidaris,
Gorgonian carnivores I,
Herbivorous gastropod carnivores I,
Herbivorous gastropods II-IV,
Holothurian detritivores,
Macroplanktonic carnivores I,
Micro-detritivores,
Molluscivores II-III,
Planktonic bacteria,
Polychaete predators (gastropods),
Seagrasses,
Sponge-anemone carnivores I,
Spongivorous nudibranchs\\
\bottomrule
\end{tabular*}
\end{table}

Let us look at the ranking that we obtained from \dtname{Nfl} dataset
using $k = 3$ groups, given in Table~\ref{tab:nfl}. We see from the results
that the obtained ranking is very sensible. 7 of 8 teams in the top group
consists of playoff teams of 2014 season, while the bottom group consists of teams that have a
significant losing record.

Finally, let us look at the rankings obtained \dtname{Reef} dataset. The graph is in fact a DAG
with 19 groups. To reduce the number of groups we rank the guilds into $k = 4$ groups.
The condensed results are given in Table~\ref{tab:reef}. We see that the top group
consists of large fishes and sharks, the second group contains mostly smaller fishes,
a large portion of the third group are crustacea, while the last group contains the bottom
of the food chain, planktons and algae. We should point out that this ranking is done purely
on food web, and not on type of species. For example, \emph{cleaner crustacea} is obviously very
different than plankton. Yet \emph{cleaner crustacea} only eats \emph{planktonic bacteria} and \emph{micro-detritivores}
while being eaten by many other guilds. Consequently, it is ranked in the bottom group.

\section{Concluding remarks}\label{sec:conclusions}

In this paper we studied the problem of discovering a hierarchy in a directed
graph that minimizes agony. We introduced several natural extensions: \emph{(i)}
we demonstrated how to compute the agony for weighted edges, and \emph{(ii)} how
to limit the number of groups in a hierarchy. Both extensions cannot be
handled with current algorithms, hence we provide a new technique by
demonstrating that minimizing agony can be solved by solving a capacitated
circulation problem, a well-known graph problem with a polynomial solution.

We also introduced a fast divide-and-conquer heuristic that produces the rankings
with competitive scores.

We should point out that we can further generalize the setup by allowing each
edge to have its own individual penalty function. As long as the penalty
functions are convex, the construction done in Section~\ref{sec:convex} can
still be used to solve the optimization problem. Moreover, we can further
generalize cardinality constraint by requiring that only a subset of vertices
must have ranks within some range. We can have multiple such constraints.

There are several interesting directions for future work. As pointed out in
Section~\ref{sec:convex} minimizing convex penalty increases the number
of edges when solving the corresponding circulation problem. However, these
edges have very specific structure, and we conjecture that it is possible
to solve the convex case without the additional computational burden.

\bibliographystyle{abbrvnat}
\bibliography{bibliography}

\appendix
\section{Proof of Proposition~\lowercase{\ref{prop:concave}}}

\begin{proof}
To prove the completeness we
will provide reduction from \mcutprb~\cite{garey1979computers}. An instance of \mcutprb
consists of an undirected graph, and we are asked to partition vertices
into two sets such that the number of cross edges is larger or equal than
the given threshold $\sigma$.

Note that the conditions of the proposition guarantee that $\pen{0} > 0$.

Assume that we are given an instance of \mcutprb, that is, an undirected graph
$G = (V, E)$ and a threshold $\sigma$. Let $m = \abs{E}$.
Define a weighted directed graph $H = (W, F, w)$ as follows. Add
$V$ to $W$.  For each edge $(u, v) \in E$, add a path with $t$ intermediate
vertices from $u$ to $v$, the length of the path is $t + 2$. Add also a path in
reversed direction, from $v$ to $u$. Set edge weights to be 1.
Add 4 special vertices $\alpha_1, \ldots, \alpha_4$. Add edges $(\alpha_{i + 1}, \alpha_i)$, for $i = 1, \ldots, 3$ 
with a weight of 
\[
	C = 2B\frac{\pen{0}}{\pen{2} - \pen{1}}, \quad\text{where}\quad  B = 2(t + 1)m\quad.
\]
Add edges $(\alpha_i, \alpha_{i + 1})$, for $i = 1, \ldots, 3$
with a weight of 
\[
	D = 4C \pen{1} / \pen{0} + B\quad.
\]
Add edges $(\alpha_0, v)$ and $(v, \alpha_4)$, for each $v \in V$,  with a weight of $D$. 

Let $r$ be the optimal ranking for $H$.
We can safely assume that $r(\alpha_1) = 0$.
We claim that $r(\alpha_i) = i - 1$, and $r(v) = 1, 2$ for each $v \in V$.
To see this, consider a ranking $r'$ such that $r'(\alpha_i) = i - 1$ and the rank for the remaining vertices is 2.
The score of this rank is 
\[
	\score{H, r'} = 3C\pen{1} + 2(t + 1)m\pen{0} = 3C\pen{1} + B\pen{0}\quad.
\]
Let $(u, v) \in F$ with the weight of $D$. If $r(u) \geq r(v)$, then the
score of $r$ is at least $D \pen{0} = 4C\pen{1} + B \pen{0}$ which is more than $\score{H, r'}$.
Hence, $r(u) < r(v)$. 
Let $(u, v) \in F$ with the weight of $C$. Note that  $r(u) \geq r(v) + 1$.
Assume that $r(u) \geq r(v) + 2$.  Then the score is at least
\[
	3C \pen{1} + C(\pen{2} - \pen{1}) = 3C\pen{1} + 2B\pen{0},
\]
which is a contradiction. This guarantees that
$r(\alpha_i) = i - 1$, and $r(v) = 1, 2$ for each $v \in V$.

Consider $(u, v) \in E$ and let $u = x_0, \ldots, x_{t + 1} = v$ be the corresponding path in $H$.
Let $d_i = r(x_i) - r(x_{i + 1})$ and set $\ell = t + r(u) - r(v)$.
Let $P = \sum \pen{d_i}$ be the penalty contributed by this path. Note that
$P \leq \pen{\ell}$, a penalty that we achieve by setting $r(x_i) = r(x_{i - 1}) + 1$ for $i = 1, \ldots, t$.
This implies that $d_i \leq \ell$.
The condition of the proposition now implies
\[
\begin{split}
	P & = \sum_{i = 0}^t \pen{d_i} = \sum_{i = 0, d_i \geq 0}^t \pen{d_i} \\
	& \geq \sum_{i = 0}^t \frac{\max(d_i + 1, 0)}{\ell + 1}\pen{\ell} \\
	& \geq \sum_{i = 0}^t \frac{d_i + 1}{\ell + 1}\pen{\ell}\\
	& = \frac{\pen{\ell}}{\ell + 1}\sum_{i = 0}^t 1 + r(x_i) - r(x_{i + 1})\\
	& = \frac{\pen{\ell}}{\ell + 1}(t + 1 + r(u) - r(v)) = \pen{\ell}\quad.\\
\end{split}
\]
This guarantees that $P = \pen{\ell}$.

Partition edges $E$ into two groups,
\[
	X = \set{(u, v) \in E \in r(u) = r(v)}
\]
and
\[
Y = \set{(u, v) \in E \in r(u) \neq r(v)}\quad.
\]
Let $\Delta = \pen{t - 1} + \pen{t + 1} - 2\pen{t}$.
Note that concavity implies that $\Delta < 0$.
Then
\[
\begin{split}
	\score{H, r} & = 3C + \abs{X}2\pen{t} + \abs{Y}(\pen{t - 1} + \pen{t + 1}) \\
	&= 3C + m2\pen{t} + \abs{Y}(\pen{t - 1} + \pen{t + 1} - 2\pen{t}) \\
	&= 3C + m2\pen{t} + \abs{Y}\Delta\quad.
\end{split}
\]
The first two terms are constant.
Consequently, $\score{H, r}$ is optimal if and only if $\abs{Y}$, the number of cross-edges
is maximal.

Given a threshold $\sigma$, define $\sigma' = 3C + m2\pen{t} + \Delta\sigma$.
Then $\score{H, r} \leq \sigma'$ if and only if there is a cut of $G$ with at
least $\sigma$ cross-edges, which completes the reduction.\qed
\end{proof}

\section{Proof of Proposition~\lowercase{\ref{prop:unique}}~and~\ref{prop:canon}}

We will prove both Propositions~\ref{prop:unique} and \ref{prop:canon} with the same proof.

\begin{proof}
Let $r^*$ be the ranking returned by \canon, and let $\pi^* = \pi - d$ be the
corresponding dual. Lemma~\ref{lem:dualupdate} states that $\pi^*$ satisfies
the slackness conditions, so it remains a solution to Problem~\ref{prb:dualuncirc}.
This implies also that $r^*$ is an optimal ranking.

To complete the proof we need to show that for any $r'$, we have $r^* \preceq r'$.
Note that this also proves that $r^*$ is a unique ranking having such property.

Let $r'$ be any optimal ranking, and let $\pi'$ be the corresponding dual.
We can assume that $\pi'(\alpha) = \pi^*(\alpha) = 0$.
To prove the result we need to show that $\pi'(v) \geq \pi^*(v)$.
We will prove this by induction over the shortest path tree $T$ from $\alpha$.
This certainly holds for $\alpha$.

Let $u$ be a vertex and let $v$ be its parent in $T$,
and let $e \in E(T)$ be the connecting edge.
Note that, by definition, $-t(e) = \pi^*(u) - \pi^*(v)$.
By the induction assumption, $\pi^*(v) \leq \pi'(v)$.

If $e$ is forward, then due to Eq.~\ref{eq:dualcond}
\[
	\pi'(u) - \pi'(v) \geq -t(e) = \pi^*(u) - \pi^*(v) \geq \pi^*(u) - \pi'(v)\quad.
\]
If $e$ is backward, then $f(v, u) > 0$. and Eq.~\ref{eq:slack} implies 
\[
	\pi'(u) - \pi'(v) = -t(e) = \pi^*(u) - \pi^*(v) \geq \pi^*(u) - \pi'(v)\quad.
\]
This completes the induction step, and the proves the proposition.
\qed
\end{proof}

\end{document}